\documentclass[12pt, a4paper, reqno]{amsart}
\usepackage{mathrsfs}
\usepackage{bbm}
\usepackage{diagbox}
\usepackage{multirow}
\usepackage{pgf,tikz}
\usepackage{tabu}
\usepackage{amsmath,amscd,amssymb,latexsym}
\usepackage[all]{xy}

\input xypic

\evensidemargin=\oddsidemargin
\DeclareMathVersion{can}
\DeclareMathAlphabet{\can}{OT1}{cmss}{m}{n}
\newtheorem{thm}{Theorem}[section]
\newtheorem{cor}[thm]{Corollary}
\newtheorem{lemma}[thm]{Lemma}
\newtheorem{prop}[thm]{Proposition}
\newtheorem{rem}[thm]{Remark}
\newtheorem{example}[thm]{Example}
\theoremstyle{definition}

\theoremstyle{fact}

\theoremstyle{conjecture}

\numberwithin{equation}{section}
\begin{document}

% ------------------------------------------------------------------------

\title[quaternary codes and related subfield codes]{quaternary linear  codes and related  binary subfield codes}

%Y. Wu is with the School of Computer Science, Nanjing University of Posts and Telecommunications, Nanjing 210023, P. R. China;  Shanghai Key Laboratory of Trustworthy Computing, East China Normal University, Shanghai, 200062, P. R. China;. (e-mail: yanshengwu@njupt.edu.cn)

\author[Y. Wu]{Yansheng Wu}
\address[Y. Wu]{School of Computer Science, Nanjing University of Posts and Telecommunications, Nanjing 210023, P. R. China;  Shanghai Key Laboratory of Trustworthy Computing, East China Normal University, Shanghai, 200062, P. R. China}
\email{yanshengwu@njupt.edu.cn}

\author[C. Li]{Chengju Li}
\address[C. Li]{ Shanghai Key Laboratory of Trustworthy Computing, East China Normal University, Shanghai, 200062, P. R. China}
\email{cjli@sei.ecnu.edu.cn}

\author[F. Xiao]{Fu Xiao}
\address[F. Xiao]{School of Computer Science, Nanjing University of Posts and Telecommunications, Nanjing 210023, P. R. China}
\email{xiaof@njupt.edu.cn}

%\author[Y. Wu]{Yansheng Wu}
%\address[]{School of Computer Science, Nanjing University of Posts and Telecommunications, Nanjing, 210023, P. R. China}
%\email{yanshengwu@njupt.edu.cn}
%\thanks{H. Dong was partially supported by the NSF under agreement DMS-1600593.}


% ----------------------------------------------------------------

\thanks{{\tiny   }}

\begin{abstract} In this paper,  we mainly study quaternary linear codes and their binary subfield codes. First we  obtain a general  explicit relationship between quaternary linear codes and their binary subfield codes in terms of generator matrices  and defining sets. Second,  we construct quaternary linear codes via simplicial complexes and  determine the weight distributions of these codes. Third, the weight distributions of the binary subfield codes of these quaternary codes are also computed by employing the general characterization.  Furthermore, we present two infinite families  of  optimal linear  codes with respect to the Griesmer Bound, and a class of binary  almost optimal  codes with respect to the Sphere Packing Bound. We also need to emphasize that we obtain at least 9 new quaternary linear codes. 

\end{abstract}

\maketitle

\section{Introduction}%\label{intro}

Let $\Bbb F_{q^{m}}$ be the ﬁnite ﬁeld with $q^{m}$ elements, where $q$ is a power of a prime and $m$ is a positive integer. Given an $[ n, k ]$ linear code $\mathcal{C}$ over $\Bbb F_{q^{m}}$, Ding and Heng \cite{DH} recently constructed a new  linear code $\mathcal{C}^{(q)}$ over  $\Bbb F_q$ with respect to $\mathcal C$,   which is called a subfield code.  In the paper, the authors mainly developed the general theory of subfield codes, investigated subfield codes of two families of ovoid codes, and presented some new and optimal subfield codes. After that,   there have been literature  on  subfield codes of combinatorial codes (\cite{HD1,HDW}); MDS codes (\cite{HD2, TWD}), and some other linear codes (\cite{HWD,WZ, WZZ, XY, ZWL}). We record a table here (Table 1) for the convenience of the reader. In the table, we list some optimal linear subfield codes  with respect to the Griesmer Bound or  the Sphere Packing Bound.

Based on the generic  construction for linear codes, Hyun {\em et al}. \cite{HLL} constructed some infinite families of  binary optimal  linear codes  by choosing the defining set  as the complement of some  simplicial complexes.  A more general situation was considered by Hyun {\em et al}. \cite{HKWY} by using posets, and they presented some optimal and minimal binary linear
 codes not satisfying the condition of Ashikhmin-Barg \cite{AB}.    Recently, Zhu and Wei \cite{ZW}  constructed quaternary linear codes via simplicial complexes and presented an infinite family of  minimal  optimal quaternary linear  codes with respect to the Griesmer bound.
 
 \begin{table}[h]
\caption{Optimal linear subfield codes }
\begin{center}
\begin{tabu} to 1\textwidth{|X[1,c]|X[0.6,c]| X[3,c]|X[0.7,c]|X[1.5,c]|X[1.2,c]|}
\hline
\tiny{Reference} &\tiny{$q$-Ary}&\tiny{$[n,k,d]$ Code}& \tiny{\#Weight}&\tiny{ Bound}& \tiny{Result}\\
\hline
\tiny{[4]}&\tiny{$q$-ary}&\tiny{$[q^{2}+1, 4, q^{2}-q]$}&\tiny{2}&\tiny{Griesmer bound} &\tiny{ Thm.1.1}\\
%\cline{2-6}
%&\tiny{$p$-ary}&\tiny{$[p^{2m}+1,3m+1, p^{2m-1}(p-1)-p^{m-1}]$}&\tiny{6}&\tiny{} &\tiny{ Thm. 4.6}\\
\hline

\multirow{2}*{\tiny{[7]}}&\tiny{binary}&\tiny{$[2^{m}+2,2^{m}-m,4]$}&\tiny{}&\tiny{Sphere Packing} &\tiny{ Thm.11}\\
\cline{2-6}
&\tiny{$p$-ary}&\tiny{$[p^{m}+1, p^{m}-2m,4]$}&\tiny{}&\tiny{Sphere Packing} &\tiny{ Thm.16}\\
\hline
\multirow{4}*{\tiny{[8]}}&\tiny{binary}&\tiny{$[2^{m}+1,2^{m}-m,3]$}&\tiny{}&\tiny{Sphere Packing} &\tiny{ Thm.VII. 4}\\
\cline{2-6}
&\multirow{3}*{\tiny{$p$-ary}}&\tiny{$[p^{m}+1, p^{m}+1-2m,3]$}&\tiny{}&\tiny{Sphere Packing} &\tiny{ Thm.VI.7}\\
\cline{3-6}
&&\tiny{$[p^{m}+1, p^{m}-m,3]$}&\tiny{}&\tiny{Sphere Packing} &\multirow{2}*{\tiny{ Thm.V.1}}\\
\cline{3-5}
&&\tiny{$[p^{m}+1, m+1, p^{m-1}(p-1)]$}&\tiny{3}&\tiny{Griesmer bound} &\\
\hline

\multirow{2}*{\tiny{[9]}}&\multirow{2}*{\tiny{binary}}&\tiny{$[2^{m+s}+2^{s}-2^{m},2^{m+s}+2^{s}-2^{m}-2m-2-m,4]$}&\tiny{}&\tiny{Sphere Packing} &\tiny{ Thm.12}\\
\cline{3-6}
&&\tiny{$[2^{m}+2, 2^{m}-2m,6]$}&\tiny{}&\tiny{Sphere Packing} &\tiny{ Thm.18}\\
\hline

\multirow{4}*{\tiny{[10]}}&\multirow{4}*{\tiny{$q$-ary}}&\tiny{$[q^{2}-1,4,q^{2}-q-2]$}&\tiny{5}&\tiny{Griesmer bound} &\multirow{2}*{\tiny{ Thm.6}}\\
\cline{3-5}
&&\tiny{$[q^{2}-1,q^{2}-5,4]$}&\tiny{}&\tiny{Sphere Packing} &\\
\cline{3-6}
&&\tiny{$[q^{2}, 4,q^{2}-q-1]$}&\tiny{4}&\tiny{Griesmer bound} &\multirow{2}*{\tiny{ Thm.9}}\\
\cline{3-5}
&&\tiny{$[q^{2}, q^{2}-4,4]$}&\tiny{}&\tiny{Sphere Packing} &\\
\hline

\multirow{4}*{\tiny{[21]}}&\multirow{4}*{\tiny{binary}}&\tiny{$[2^{2m}+2, 2m+1, 2^{2m-1}]$}&\tiny{4}&\tiny{Griesmer bound} &\multirow{2}*{\tiny{ Thm.1}}\\
\cline{3-5}
&&\tiny{$[2^{2m}+2, 2^{2m}-2m+1, 3]$}&\tiny{}&\tiny{Sphere Packing} &\\
\cline{3-6}
&&\tiny{$[2^{2m}+1, 2m+1, 2^{2m-1}]$}&\tiny{3}&\tiny{Griesmer bound} &\multirow{2}*{\tiny{ Thm.2}}\\
\cline{3-5}
&&\tiny{$[2^{2m}+1, 2^{2m}-2m, 3]$}&\tiny{}&\tiny{Sphere Packing} &\\
\hline
\end{tabu}
\end{center}
\end{table}

From the argument of Xiang in \cite{X} and the generic  construction for linear codes, each  linear code  can be   expressed as the defining code $\mathcal{C}_{D}$  with a defining set $D$.   Due  to the key role in estimating the error-correcting capability of  codes, weight distributions is  an  important research topic in coding theory. Motivated by the above work, we ask the following questions:
 
{\it 1)  What is the relationship between  quaternary linear codes  and their binary subfield codes ? Namely for the generic construction of linear codes, what  is the relationship between defining sets of   quaternary linear  codes and the binary subfield codes ?
 
 2)  Can we obtain more optimal quaternary linear codes and determine their weight distributions ?
 
 3) Can we obtain more optimal the binary subfield  linear codes of quaternary linear codes and determine their weight distributions ?
 }
 %Notice that linear codes from the  generic construction via posets are all over prime fields and  the main difficulty   is to calculate the frequencies of their codewords.  It seems that new techniques are required to go beyond prime fields. In this paper, we will provide such a technique for linear codes over the finite field $\Bbb F_4$. 
 
 The basic questions above are the major motivation of this paper. First, we find a direct approach to deal with the relationship between  quaternary linear codes  and their binary subfield codes.  Second, we will follow the idea in \cite{HLL}, and use simplicial complexes to construct  quaternary linear codes. Weight distributions of  these quaternary codes are  determined when those simplicial complexes are generated by one or two maximal elements.   Third, we  also compute  weight distributions of the binary subfields code
 of these quaternary codes.   In addition, we present several classes of optimal and almost optimal linear  codes and some examples of linear codes with optimal parameters. {\color{blue} By Magma,  we obtain at least 9 new quaternary linear codes. 
 }
 
 The rest of this paper is organized as follows. In Section 2, we  recall some bounds on linear codes, some  concepts of  simplicial complexes, and generating functions. In Section 3,  we will deal with the relationship between quaternary linear codes and their binary subfield codes.  We will compute the weight distributions of some quaternary codes and their binary subfield codes in Sections 4 and  5. Furthermore, we obtain two classes of optimal linear codes and a class of almost optimal binary codes. In Section 6, we  present  some new quaternary linear codes and  conclude the paper.

 %we determine the weight distributions of these quaternary codes and find a class of minimal quaternary linear codes. %As applications, we compute the minimum distance of the dual of these quaternary codes and find many (almost) optimal and minimal codes.

%Wu and Lee \cite{W2} first used the
 %difference of simplicial complexes for construction of binary linear complementary dual codes and binary self-orthogonal codes.

%\section{Problem formulation and main results}\label{problem results}
%\subsection{Problem formulation}\label{prob formulation}

%In this paper, we aim to establish

\section{Preliminaries}

 \subsection{Two bounds of linear codes}
 
 \

Let $\mathcal{C}$ be an $[n,k,d]$ linear code over $\Bbb F_q$.
Assume that there are $A_i$ codewords in $\mathcal C$ with Hamming weight $i$ for $1\le i \le n$.
Then $\mathcal C$ has weight distribution $(1, A_1, \ldots, A_n)$ and  weight enumerator $1+A_1z+\cdots +A_nz^n$.
 Moreover, if the number of nonzero $A_{i}$'s in the sequence $(A_1, \ldots, A_n)$ is exactly equal to $t$,
 then the code is called {\it $t$-weight}. The $[n,k,d]$ code  $\mathcal{C}$ is called {\it distance optimal} if there is no $[n,k,d+1]$ code (that is, this code has the largest minimum distance for given length $n$ and dimension $k$),
 and it is called {\it almost optimal} if an $[n, k, d + 1]$ code is  distance optimal (refer to~\cite[Chapter 2]{HP}). 
 
 Next we recall two well-known bounds on linear codes.
 \begin{lemma}{\rm  (Griesmer Bound \cite{G})   For a given  $[n, k, d]$  linear code over $\Bbb F_q$,  there is a bound as follows:
 $$ \sum_{i=0}^{k-1}\bigg\lceil {\frac{d}{q^i}} \bigg\rceil \le n,  $$ where $\lceil {\cdot} \rceil$ is the ceiling function. 
 }
 \end{lemma}
 We say that a linear code is a Griesmer code if it meets the Griesmer bound with equality. One can verify that Griesmer codes are distance optimal.
 
  \begin{lemma}{\rm  (Sphere Packing Bound \cite{HP})   For a given  $[n, k, d]$  linear code over $\Bbb F_q$,  there is a bound as follows:
 $$\sum_{i=0}^{\left \lfloor \frac{d-1}{2} \right \rfloor} \dbinom{n}{i} (q-1)^{i}\le q^{n-k}, $$
  where $ \left \lfloor {\cdot}  \right \rfloor$ is the floor function. 
  
 }
 \end{lemma}
 
 When we have a code for which equality in the above bound is ture, the code is called {\it perfect}. One can verify that perfect codes are also distance optimal.
 
 \subsection{The generic construction of linear codes}
 
 \

 Let $m$ be a positive integer, $q$ be a {\color{blue} prime power}, and $(V_m,\cdot)$ be an $m$-dimensional vector space over $\Bbb F_q$, where $\cdot$ denotes an  inner product on $V_m$. For a linear code of length $n$ over $\Bbb F_q$, there is a generic construction as follows:
\begin{equation}\label{eq2.1}
\mathcal C_D=\{(x\cdot d_1, x\cdot d_2, \ldots, x\cdot d_n): x\in V_m\}\end{equation}
where $D=\{d_1, \ldots, d_n\} \subseteq V_m$. The set $D$ is called the defining set of the code $\mathcal{C}_{D}.$ If the set $D$ is properly chosen, the code $\mathcal{C}_{D}$ may have good parameters. The following two situations are common:

(1) When $V_m=\Bbb F_{q^m}$, $x\cdot y=\mbox{Tr}_{q^{m}/q}(xy)$ for $x,y\in \Bbb F_{q^m}$ and $\mbox{Tr}_{q^{m}/q}$ is the  trace function from $\Bbb F_{q^m}$ to $\Bbb F_{q}$. In this case, the corresponding code $\mathcal C_D$ in  Equation \eqref{eq2.1} is called a trace code over $\Bbb F_q$. This generic construction was first introduced by Ding {\em et al}. \cite{D1}. %Many known codes have been produced by selecting a proper defining set, see \cite{HY1, LYL} for examples.   Note that defining sets here are almost all related with trace functions, and the computations of weight distributions of corresponding linear codes are heavily dependent on  known results of exponential sums.

(2) When $V_m=\Bbb F_q^m$, ${\bf  x}\cdot {\bf  y}=\sum_{i=1}^mx_iy_i$ for ${\bf  x}=(x_1, \ldots, x_m), {\bf  y}=(y_1, \ldots, y_m)\in \Bbb F_q^m.$  This standard construction in  Equation \eqref{eq2.1} can be also found in  \cite{HP}.  %Recently, Zhou et al. \cite{ZLTD} investigated four infinite families of binary linear codes and   obtained some binary  linear complementary dual or self-orthogonal  codes  based on the above generic construction

\subsection{Simplicial complexes and generating functions}$~$
\

Let $\Bbb F_2$ be the finite field with  two elements. Assume that $m$ is a positive integer. The support $\mathrm{supp}({\bf v})$ of a vector ${\bf v} \in \Bbb F_2^m$ is defined by the set of nonzero coordinates. The Hamming weight $wt({\bf v})$ of ${\bf v}\in \mathbb{F}^m_2$ is defined by the size of $\mathrm{supp}({\bf v})$.  For two subsets $A, B\subseteq [m]$, the set $\{x: x\in A\mbox{ and } x\notin B\}$ and the number of elements in the set $A$ are denoted by
$A\backslash B$ and $|A|$, respectively.

For two vectors ${\bf u,v}\in \mathbb{F}_2^m$, we say ${\bf v}\subseteq {\bf u }$ if $\mathrm{supp}({\bf v})\subseteq \mathrm{supp}({\bf u})$.  We say that a family $\Delta \subseteq \mathbb{F}_2^m $ is a {\it simplicial complex} if ${\bf  u}\in \Delta$ and ${\bf v}\subseteq {\bf u}$ imply ${\bf v}\in \Delta$. For a simplicial complex $\Delta$, a maximal element of $\Delta $ is one that is not properly contained in any other element of $\Delta$. Let $\mathcal{F}=\{F_1, \ldots, F_l\}$ be the family of maximal elements of $\Delta$. For each $F\subseteq [m]$, the simplicial complex $\Delta_F$ generated by $F$ is defined to be the family of all subsets of $F$.

Let $X$ be a subset of $\mathbb{F}_2^m$. Hyun {\em et al}. \cite{CH} introduced the following $m$-variable generating function associated with the set $X$:
$$\mathcal{H}_{X}(x_1,x_2,\ldots, x_m)=\sum_{\mathbf u\in X}\prod_{i=1}^mx_i^{u_i}\in \mathbb{Z}[x_1,x_2, \ldots, x_m],
$$
where ${\bf u}=(u_1,u_2,\ldots, u_m)\in \mathbb{F}_2^m$ and $\mathbb{Z}$ is the ring of integers. %We observe that
 %$\mathcal{H}_{\emptyset}(x_1,x_2, \ldots, x_m)=0.$

%$(2)$ $\mathcal{H}_{X}(x_1,x_2, \ldots, x_m)+\mathcal{H}_{X^{c}}(x_1,x_2, \ldots, x_m)=\mathcal{H}_{\mathbb{F}_2^m}(x_1,x_2, \ldots, x_m)=\prod_{i\in[m]}(1+x_i)$.

%We now present the closed form of the generating function associated with an order ideal of $\mathcal{O}_{\mathbb{P}}$.
%It allows us to compute efficiently the Hamming weights of linear codes defined in Section 5.

The following lemma plays an important role in determining the weight distributions of the quaternary codes defined in Equation  \eqref{eq2.1}.

\begin{lemma}\cite[Theorem 1]{CH}  \label{lem2.1} {\rm 
Let $\Delta$ be a simplicial complex of $\mathbb{F}_2^m$ with the set of maximal elements $\mathcal {F}$. Then
 \begin{align*}
 \mathcal{H}_{\Delta}(x_1,x_2,\ldots, x_m)=\sum_{\emptyset\neq S\subseteq \mathcal{F}}(-1)^{|S|+1}\prod_{i\in \cap S}(1+x_i),
 \end{align*}
 where $\cap S$ denotes the intersection of all elements in $S$.    
In particular, we also have $$|\Delta|=\sum_{\emptyset\neq S\subseteq \mathcal{F}}(-1)^{|S|+1}2^{|\cap S|}.$$ }
\end{lemma}

There is a bijection between $\mathbb{F}_2^m$ and $2^{[m]}$ being the power set of $[m]=\{1, \ldots, m\}$, defined by ${\bf v}\mapsto$ supp$({\bf v})$.  Throughout this paper, we will identify a vector in $\mathbb{F}_2^m$  with its support. 

\begin{example} {\rm Let $\Delta$ be a simplicial complex of $\mathbb{F}_2^4$ with the set of maximal elements $\mathcal {F}=\{(1,1,0,0), (0,1,1,0),(0,0,1,1)\}$.  Then 
\begin{align*}
 \mathcal{H}_{\Delta}(x_1,x_2,
 x_3, x_4)=\prod_{i\in \{1,2\}}(1+x_i)+\prod_{i\in \{3,4\}}(1+x_i)+\prod_{i\in \{2,3\}}(1+x_i)-(1+x_{2})-(1+x_{3})\\
 =1+x_1+x_2+x_3+x_4+x_1x_2+x_2x_3+x_3x_4.
 \end{align*}
 and $|\Delta|=8.$
}
\end{example}

\section{Relationship between quaternary codes and the subfield codes}

%In the paper \cite{WZY}, the authors first constructed linear codes over the finite ring $\Bbb F_2+u\Bbb F_2$ with $u^2=0$ and obtained  many optimal binary linear codes by Gray map. After that, Wu et al. \cite{WH} also considerd the case of $\Bbb F_p+u\Bbb F_p$ with $u^2=0$ and $p$ is an odd prime number.
 %Let $\mathbb Z_4$ be the ring of integers modulo $4$. For each $u \in \mathbb Z_4$ there is a unique representation $u =
%a + 2b$, where $a, b \in \Bbb F_2$. Here the element $2$ in $\mathbb Z_4$  plays a similar role, which like $u$ for the ring $\Bbb F_2+u\Bbb F_2$, and the only difference is the characteristics of the two rings.

For the finite field $\Bbb F_4$, as we known $\Bbb F_4\cong \Bbb F_2[x]/\langle x^2+x+1\rangle$, where $x^2+x+1$ is the only irreducible polynomial of degree two in $\Bbb F_2[x]$. Let $w$ be an element in some extend field of $\Bbb F_2$ such that $w^2+w+1=0$. Then $\Bbb F_4=\Bbb F_2(w)$ and for each $u \in \mathbb F_4$ there is a unique representation $u =
a + wb$, where $a, b \in \Bbb F_2$. Let $m$ be a positive
integer, and $\mathbb F_4^m$
be the set of $m$-tuples over $\mathbb F_4$.  Any vector $\mathbf x \in  \mathbb F_4^m$
can be written as $\mathbf x =\mathbf a + w\mathbf b$, where
$\mathbf a,\mathbf b\in\mathbb F_2^m $.

From the argument of Xiang in \cite{X}, any quaternary  linear code of length $n$ can be  also expressed as the code $\mathcal{C}_{D}$ in Equation \eqref{eq2.1}, where $D = \{ \mathbf d_{1}, \mathbf d_{2} , \ldots , \mathbf d_{n} \} \subseteq \Bbb F_{4}^{m} $ and $m$ is some positive integer. %By the above argument, there are two subsets $D_{1}, D_{2}$ of $\mathbb F_2^m$ such that $D=D_{1}+wD_{2}$.  ???

The following result plays an important role in the research of the subfield codes. 

\begin{lemma} \cite[Theorem 2.4]{DH} \label{lem3.1}
{\rm Let $\mathcal C$ be an $[n,k]$ linear code over $\mathbb{F}_{q^m}$ with generator matrix 
$$G=\left(\begin{array}{cclc}  g_{11} & g_{12} &\ldots  & g_{1n}\\  g_{21} & g_{22} &\ldots &  g_{2n}\\ \vdots &\vdots &\ddots&\vdots \\  g_{k1} & g_{k3} &\ldots &  g_{kn}\end{array}\right).$$
Let $\{\alpha_{1}, \alpha_{2}, \ldots, \alpha_{m}\}$ be a basis of $\mathbb{F}_{q^m}$ over $\mathbb{F}_{q}$. Then the subfield code $\mathcal C^{(q)}$ with respect to $\mathcal C$ has a generator matrix 
$$G^{(q)}=\left(\begin{array}{cclc} G_{1}^{(q)}\\ 
G_{2}^{(q)} \\ 
\vdots  \\ 
G_{k}^{(q)} \end{array}\right),$$
where each $G_{i}^{(q)}$ is defined as $$\left(\begin{array}{cclc}\mbox{Tr}_{q^{m}/q}( g_{i1} \alpha_{1})&\mbox{Tr}_{q^{m}/q}(  g_{i2} \alpha_{1}) &\ldots  &\mbox{Tr}_{q^{m}/q}(  g_{in} \alpha_{1})\\ \mbox{Tr}_{q^{m}/q}(  g_{i1} \alpha_{2}) &\mbox{Tr}_{q^{m}/q}(  g_{i2} \alpha_{2}) &\ldots & \mbox{Tr}_{q^{m}/q}( g_{in} \alpha_{2}) \\ \vdots &\vdots &\ddots&\vdots \\ \mbox{Tr}_{q^{m}/q}( g_{i1} \alpha_{m})  &\mbox{Tr}_{q^{m}/q}(  g_{i2} \alpha_{m})  &\ldots & \mbox{Tr}_{q^{m}/q}( g_{in} \alpha_{m}) \end{array}\right).$$

}
\end{lemma}

In  Lemma \ref{lem3.1}, let $q=2$, $m=2$ and  $\{\alpha_{1}=1,\alpha_{2}=w\}$ be a basis of $\Bbb F_{4}$ over $\Bbb F_{2}$.  Assume that  $g_{ij}=g_{ij}^{(0)}+wg_{ij}^{(1)}$, where $g_{ij}^{(0)}, g_{ij}^{(1)}\in \Bbb F_{2}.$  Hence  $ \mbox{Tr}_{4/2}( g_{ij} \alpha_{1})=g_{ij}^{(1)}$ and $ \mbox{Tr}_{4/2}( g_{ij} \alpha_{2})=g_{ij}^{(0)}+g_{ij}^{(1)}$. Then we have the following theorem. 

\begin{thm}\label{thm3.2} {\rm Let $\mathcal C$ be an $[n,k]$ linear code over $\mathbb{F}_{4}$ with generator matrix $G=G_{1}+wG_{2}$, where  $w\in \Bbb F_{4}$ with $w^{2}+w+1=0$ and $G_{1}, G_{2}$ are two matrices  over $\Bbb F_{2}$.  Then the binary  subfield code $\mathcal C^{(2)}$ with respect to $\mathcal C$ has a generator matrix  $$G^{(2)}=\left(\begin{array}{cclc} G_{2}\\ 
G_{1}+G_{2} \end{array}\right).$$ Moreover, if the quaternary code  $\mathcal C$ has the defining set $D=D_{1}+wD_{2}$ with $D_{1}, D_{2} \subseteq \mathbb F_2^m$,  then the binary  subfield code $\mathcal C^{(2)}$ with respect to $\mathcal C$ has  defining set: $$D^{(2)}=\{(\mathbf d_{2}, \mathbf d_{1}+\mathbf d_{2}): \mathbf d_{1}\in D_{1}, \mathbf d_{2}\in D_{2} \}.$$

}
\end{thm}

{\color{blue} \begin{rem} {\rm There is a well-known  the  Plotkin construction,  for linear codes from old codes which is documented in  \cite{MS}. By Theorem 3.2, the subfield construction of quaternary codes includes  the  Plotkin construction.
}
\end{rem}

}

We give the following example to illustrate Theorem 3.2.

\begin{example}{\rm  Let $\mathcal C$ be a $[4,2]$ linear code over $\mathbb{F}_{4}$ with defining set $D=D_{1}+wD_{2}$, where $D_{1}=\{(0,1), (1,0)\}$ and $D_{2}=\{(0,1), (1,1)\}$.  Then its generator matrix is $$\left(\begin{array}{cclc} w & 0 &1+w  & 1\\  1+w & 1+w &w&  w\end{array}\right)=\left(\begin{array}{cclc} 0 & 0 &1  & 1\\  1 & 1 &0&  0\end{array}\right)+w\left(\begin{array}{cclc} 1 & 0 &1  & 0\\  1 & 1 &1&  1\end{array}\right).$$
By Lemma \ref{lem3.1}, the binary subfield code $\mathcal C^{(2)}$ with respect to $\mathcal C$ has a generator matrix 
$$G^{(2)}=\left(\begin{array}{cclc}1 & 0 &1  & 0\\  1 & 1 &1&  1\\ 
1 & 0 &0  & 1\\  0 & 0 &1&  1\\ 
 \end{array}\right).$$ By Theorem \ref{thm3.2}, $D^{(2)}=\{(1,1,1,0), (0,1,0,0), (1,1,0,1),(0,1,1,1)\}$.
}
\end{example}

%\begin{thm}\label{thm3.3} {\rm Let $\mathcal C$ be an $[n,k]$ linear code over $\mathbb{F}_{4}$ with defining set $D=D_{1}+wD_{2}$, where $w\in \Bbb F_{4}$ and $w^{2}+w+1=0$.  Then the binary  subfield subcode $\mathcal C| \Bbb F_{2}^{n}$ has the following defining set:$$\{(\mathbf d_{1}, \mathbf d_{2}): \mathbf d_{1}\in D_{1}, \mathbf d_{2}\in D_{2} \}.$$

%}
%\end{thm}

%We give the following examples to illustrate Theorem 3.2.

%\subsection{Minimal linear codes}$~$
%%%%%%%%%%%%%%%%%%%%%%%%%%%%%%%%%%%%%%%%%%%%%

%The following lemma is useful in finding a minimal linear code violating the condition of Aschikhmin-Barg.
%\begin{lemma}\cite[Corollary 3.2]{HDZ}
%Let $\mathcal{C}$ be a two-weight  linear code over $\mathbb{F}_q$ with nonzero weights $w_1$ and $w_2$. Then the code  $\mathcal{C}$ is minimal provided that $$jw_1\neq (j-1)w_2$$ for any integer $j$ with $2\le j\le q$.
%\end{lemma}

\section{ Weight distributions of quaternary codes}
%%%%%%%%%%%%%%%%%%%%%%%%%%%%%%%%%%%%%%%%%%%%%%%%%%%%%%%%%%%%

%In this section, we determine the weight distributions of two types of linear codes defined in (5.1) and (5.5) below which are involved with hierarchical posets of two levels.

%\subsection{Linear codes from hierarchical posets with two levels}$~$

%A generic construction of binary linear codes from subsets of $\Bbb F_{2^m}$ was considered in \cite{D1}. Recently, Zhou et al. \cite{ZTLD} presented a view point of vector space of the generic construction. Inspired by this idea, we introduce a construction of binary linear codes from posets. {\bf This sentence should be deleted because $(5.1)$ is an usual linear code which can be expressed by the generating matrix}

%Let $D=\{g_1,g_2,\cdots, g_n\}\subseteq \Bbb F_{2}^m$. A binary linear code of length $n$ is defined by
%\begin{equation}\mathcal{C}_D=\{(x\cdot g_1, x\cdot g_2, \cdots, x\cdot g_n): x\in \Bbb F_2^n\}.\end{equation}
%%The set $D$ is called the defining set of the code $\mathcal{C}_D$. Clearly, $\mathcal{C}_D$ is the linear code generated by the row vectors of the matrix $G=[g_1g_2\cdots g_n]$. Therefore $\mathcal{C}_D$ is an $[n,k]$ linear code with $k =
%Rank(G)$, which is  the rank of the matrix $G$.

In this section, we will construct some quaternary codes via  simplicial complexes and determine their  weight distributions.

Let $D_{1},D_{2}$ be two  subsets of $\mathbb{F}_2^m$ and $D=D_{1}+wD_{2}\subseteq \Bbb F_4^{m*}$, where  $\Bbb F_4^{m*}$ is the set of non-zero element of  $\Bbb F_4^{m}$ and $w\in\Bbb F_4$ such that $w^2+w+1=0$.
We define a quaternary code as follows:
\begin{equation}\label{eq2}
\mathcal C_D=\{c_{D}(\mathbf{a})=(\mathbf{a}\cdot \mathbf{d})_{\mathbf{d}\in D}:\mathbf a\in \Bbb F_4^m\}.\end{equation}
First of all, from Equation \eqref{eq2}, it is easy to check that  the code $\mathcal{C}_D$ is a  quaternary linear code. The length of the code $\mathcal{C}_{D}$ is $|D|$. If $\mathbf{a}=\mathbf{0}$, then the Hamming weight of the codeword $c_{D}(\mathbf{a})$ is equal to  $\mbox{wt}(c_{D}(\mathbf{a}))=0$. Next we  assume that $\mathbf{a}\neq \mathbf{0}$.
Suppose that $\mathbf{a}=\boldsymbol{\alpha}+w\boldsymbol{\beta}$, $\mathbf{d}={\bf d}_1+w{\bf d}_2$, where ${\boldsymbol{\alpha}=(\alpha_1, \ldots, \alpha_m),\boldsymbol{\beta}=(\beta_1, \ldots, \beta_m)}\in \Bbb F_2^m$, ${\mathbf d_1}\in D_{1}$, and ${\mathbf d_2}\in D_{2}$.   Then
\begin{eqnarray*}\mbox{wt}(c_{D}(\mathbf{a}))
&=&\mbox{wt}((( \boldsymbol{\alpha}+w\boldsymbol{\beta})\cdot({\bf d}_1+w{\bf d}_2)) _{{\bf d}_1\in  D_{1}, {\bf d}_2\in D_{2}})\nonumber\\
&=&\mbox{wt}((\boldsymbol{\alpha}{\bf d}_1+w(\boldsymbol{\alpha}{\bf d}_2+\boldsymbol{\beta}{\bf d}_1)+w^2\boldsymbol{\beta}{\bf d}_2)_{{\bf d}_1\in  D_{1}, {\bf d}_2\in  D_{2}})\nonumber\\
&=&\mbox{wt}((\boldsymbol{\alpha}{\bf d_1}+\boldsymbol{\beta}{\bf d}_2+w(\boldsymbol{\beta}{\bf d}_2+\boldsymbol{\alpha}{\bf d}_2+\boldsymbol{\beta}{\bf d}_1))_{{\bf d}_1\in  D_{1}, {\bf d}_2\in  D_{2}}).
\end{eqnarray*}
By the definition of Hamming weight of vector $\mathbf{x}=\mathbf y+w\mathbf z\in \Bbb F_4^m$ with $\mathbf{y,z}\in \Bbb F_2^m$, $\mbox{wt}(\mathbf x)=0$ if and only if $\mathbf y=\mathbf z=\mathbf0$. Hence
\begin{eqnarray*} 
\mbox{wt}(c_{D}(\mathbf{a}))
&=&|D|-\sum_{\mathbf {d}_1\in  D_{1}}\sum_{\mathbf d_2\in D_{2}}(\frac12 \sum_{y\in\Bbb F_2} (-1)^{(\boldsymbol{\alpha}\mathbf{d}_1+\boldsymbol{\beta}\mathbf{d}_2)y})(\frac12 \sum_{z\in\Bbb F_2} (-1)^{(\boldsymbol{\alpha}{\mathbf d_2}+\boldsymbol{\beta}{(\mathbf d_1+\mathbf d_2)})z})\nonumber\\
&=&|D|-\frac14\sum_{\mathbf d_1\in  D_{1}}\sum_{\mathbf d_2\in  D_{2}}(1+(-1)^{\boldsymbol{\alpha}{\mathbf d_1}+\boldsymbol{\beta}{\mathbf d_2}})(1+(-1)^{\boldsymbol{\alpha}{\mathbf d_2}+\boldsymbol{\beta}{(\mathbf d_1+\mathbf d_2)}})\nonumber\\
&=&\frac 3 4|D|-\frac 14(\sum_{\mathbf d_1\in D_{1}} (-1)^{\boldsymbol{\alpha}\mathbf d_1})(\sum_{\mathbf d_2\in D_{2}}(-1)^{\boldsymbol{\beta}{\mathbf d_2}})\nonumber\\
&-&\frac 14(\sum_{\mathbf d_1\in  D_{1}} (-1)^{\boldsymbol{\beta}\mathbf d_1})(\sum_{\mathbf d_2\in  D_{2}}(-1)^{\boldsymbol{(\alpha+\beta)}{\mathbf d_2}})\nonumber\\
&-&\frac 14(\sum_{\mathbf d_1\in  D_{1}} (-1)^{\boldsymbol{(\alpha+\beta)}\mathbf d_1})(\sum_{\mathbf d_2\in  D_{2}}(-1)^{\boldsymbol{\alpha}{\mathbf d_2}}).
\end{eqnarray*}
For  a subset $P$ of $\Bbb F_{2}^{m}$ and $\mathbf u\in \Bbb F_{2}^{m}$, let us define $\chi_{\bf u}(P)=\sum_{\mathbf v\in P}(-1)^{\bf uv}. $ Then   
\begin{align} \label{eq4}\mbox{wt}(c_{D}(\mathbf{a}))=\frac 3 4|D|\nonumber\\
-\frac 14[ \chi_{\boldsymbol{\alpha} }(D_{1})\chi_{\boldsymbol{\beta} }(D_{2}) + \chi_{\boldsymbol{\beta} }(D_{1})\chi_{\boldsymbol{\alpha+\beta} }(D_{2})+ \chi_{\boldsymbol{\alpha+\beta} }(D_{1})\chi_{\boldsymbol{\alpha} }(D_{2})].\end{align}
Let $D^{c}=\Bbb F_4^{m*} \backslash D$ and $\delta$ be the Kronecker delta function. Then \begin{align*}\mbox{wt}(c_{D^{c}}(\mathbf{a}))=\frac 3 4(|D^{c}|-2^{2m}\delta_{\mathbf{0},\mathbf{a}})
\nonumber\\+\frac 14[ \chi_{\boldsymbol{\alpha} }(D_{1})\chi_{\boldsymbol{\beta} }(D_{2}) + \chi_{\boldsymbol{\beta} }(D_{1})\chi_{\boldsymbol{\alpha+\beta} }(D_{2})+\chi_{\boldsymbol{\alpha+\beta} }(D_{1})\chi_{\boldsymbol{\alpha} }(D_{2})]\nonumber.\end{align*}
By Equation \eqref{eq4},  we have  \begin{align}  \label{eq4.5} \mbox{wt}(c_{D^{c}}(\mathbf{a}))=3\times2^{2m-2}(1-\delta_{\mathbf{0},\mathbf{a}})- \mbox{wt}(c_{D}(\mathbf{a})).\end{align}
%Hence we have the following theorem.

Recall that there is a bijection between $\mathbb{F}_2^m$ and $2^{[m]}$ being the power set of $[m]=\{1, \ldots, m\}$, defined by ${\bf v}\mapsto$ supp$({\bf v})$.   For a subset $D\subseteq \Bbb F_{4}^{m}$,  we use $D^{*}$ to denote the set $D\backslash \{\bf 0\}$.

From the proof of \cite[Theorem 5.3]{HKWY}, we derive  the following lemma, which is need in computing  weight distributions of the codes.

\begin{lemma}\label{lem4.1}\cite{HKWY} {\rm
For two subsets $A, B$ of $[m]$, we set
$$\mathcal{U}_1=\{u\in \mathbb{F}_2^m: u\cap (A\cup B)=\emptyset\},$$ $$\mathcal{U}_2=\{u\in \mathbb{F}_2^m: u\cap A=\emptyset, u\cap(B\backslash A)\neq\emptyset\},$$
$$\mathcal{U}_3=\{u\in \mathbb{F}_2^m: u\cap B=\emptyset, u\cap(A\backslash B)\neq\emptyset\},$$
$$\mathcal{U}_4=\{u\in \mathbb{F}_2^m: u\cap(A\backslash B)\neq\emptyset, u\cap(A\cap B)=\emptyset, u\cap (B\backslash A)\neq \emptyset\},$$
$$\mathcal{U}_5=\{u\in \mathbb{F}_2^m: u\cap(A\cap B)\neq\emptyset \}.$$ Then we have 
$$|\mathcal{U}_1|=2^{m-|A\cup B|}, |\mathcal{U}_2|=2^{m-|A|}-2^{m-|A\cup B|}=2^{m-|A\cup B|}(2^{|B\backslash A|}-1), $$
$$|\mathcal{U}_3|=2^{m-|B|}-2^{m-|A\cup B|}=2^{m-|A\cup B|}(2^{|A\backslash B|}-1),$$
$$|\mathcal{U}_4|=2^{m-|A\cup B|}(2^{|A\backslash B|}-1)(2^{|A\backslash B|}-1),$$
$$|\mathcal{U}_5|=2^{m-|A\cap B|}(2^{|A\cap B|}-1).$$}
\end{lemma}

\begin{prop}\label{prop4.2} {\rm Let $A,B$ be  two  subsets of $[m]$ and $D=\Delta_A+w\Delta_B\subset \Bbb F_4^m$. Then $\mathcal{C}_{D^{*}}$ in Equation \eqref{eq2} is a $[2^{|A|+|B|}-1, |A\cup B|, 2^{|A|+|B|-1}]$ quaternary code and its weight distribution
 is presented in Table 2.  \begin{table}[h]
\caption{Weight distribution of the code in Proposition \ref{prop4.2}}
\begin{center}
\begin{tabu} to 0.7\textwidth{X[1,c]|X[2,c]}
\hline
Weight &Frequency\\
\hline
$0$&$1$\\
\hline
$2^{|A|+|B|-1}$&$3(2^{|A\backslash B|+|B\backslash A|}-1)$\\
\hline
$3\times2^{|A|+|B|-2}$&$4^{|A\cup B|}-1-3(2^{|A\backslash B|+|B\backslash A|}-1)$\\
 \hline

\end{tabu}
\end{center}
\end{table}

}

\end{prop}

\begin{proof} It is easy to check that the length of the  code $\mathcal{C}_{D^{*}}$ is $|D^{*}|=2^{|A|+|B|}-1$.  To compute the weights and frequencies, we need to introduce  the following notation.

For $X$ a subset of $\mathbb{F}_2^m$, we use $\psi(\mathbf {u}|X)$ to denote a Boolean function in $m$-variable, and $\psi(\mathbf {u}|X)=1$ if and only if $\mathbf {u}\bigcap X=\emptyset$. For a vector $\mathbf {u}=(u_1, \ldots, u_m)\in \Bbb F_2^m$ and a nonempty simplicial complex $\Delta_A$, by Lemma \ref{lem2.1} we have 
\begin{multline}\label{eq4.6}
\sum_{\mathbf x\in \Delta_A} (-1)^{\mathbf {u}\cdot\bf x}=\mathcal{H}_{\Delta_A}((-1)^{u_1},(-1)^{u_2},\ldots, (-1)^{u_m})=\prod_{i\in A}(1+(-1)^{u_i})\\
=\prod_{i\in A}(2-2u_i)=2^{|A|}\prod_{i\in A}(1-u_i)=2^{|A|}\psi(\mathbf {u}|A).
\end{multline}
Suppose that $\mathbf{a}=\boldsymbol{\alpha}+w\boldsymbol{\beta}$.
%$\mathbf{d}={\bf d_1}+w{\bf d_2}$, where ${\boldsymbol{\alpha}=(\alpha_1, \cdots, \alpha_m),\boldsymbol{\beta}=(\beta_1, \cdots, \beta_m)}\in \Bbb F_2^m$, ${\bf d_1}\in D_{1}$, and ${\bf d_2}\in D_{2}$.  
By Equations  \eqref{eq4} and \eqref{eq4.6}
\begin{eqnarray*} \label{eq7}
&&\mbox{wt}(c_{D^{*}}(\mathbf{a}))=\mbox{wt}(c_{D^{}}(\mathbf{a}))= \frac 3 4|D^{}|\nonumber\\
&&-2^{|A|+|B|-2}[\psi(\boldsymbol{\alpha}|A)\psi(\boldsymbol{\beta}|B)+\psi(\boldsymbol{\beta}|A)\psi(\boldsymbol{\alpha+\beta}|B)+\psi(\boldsymbol{\alpha+\beta}|A)\psi(\boldsymbol{\alpha}|B)].\end{eqnarray*}
Let $T=\psi(\boldsymbol{\alpha}|A)\psi(\boldsymbol{\beta}|B)+\psi(\boldsymbol{\beta}|A)\psi(\boldsymbol{\alpha+\beta}|B)+\psi(\boldsymbol{\alpha+\beta}|A)\psi(\boldsymbol{\alpha}|B)$. We divide the proof into the following cases:

(1) $T=3$. In this case we have $\mbox{wt}(c_{D^{*}}(\mathbf{a}))=0$
 and $$\psi(\boldsymbol{\alpha}|A)=\psi(\boldsymbol{\alpha}|B)= \psi(\boldsymbol{\beta}|A)=\psi(\boldsymbol{\beta}|B)=\psi(\boldsymbol{\alpha+\beta}|A)=\psi(\boldsymbol{\alpha+\beta}|B)=1,$$ which is equivalent  to $$\boldsymbol{\alpha}\cap (A\cup B)=\emptyset\mbox{ and }\boldsymbol{\beta}\cap (A\cup B)=\emptyset.$$ By Lemma \ref{lem4.1}, the number of such $\mathbf{a}=\boldsymbol{\alpha}+w\boldsymbol{\beta}$ is $4^{m-|A\cup B|}.$

(2) $T=2$. Without loss of generality, suppose that $\psi(\boldsymbol{\alpha}|A)\psi(\boldsymbol{\beta}|B)=0$.
We have  \begin{eqnarray*}\left\{\begin{array}{llll}
\psi(\boldsymbol{\beta}|A)\psi(\boldsymbol{\alpha+\beta}|B)=1\\
\psi(\boldsymbol{\alpha+\beta}|A)\psi(\boldsymbol{\alpha}|B)=1
\end{array}\right.\Longleftrightarrow\left\{\begin{array}{llll}
\boldsymbol{\beta}\cap A=(\boldsymbol{\alpha+\beta})\cap B=\emptyset\\
\boldsymbol{\alpha}\cap B=(\boldsymbol{\alpha+\beta})\cap A=\emptyset.\\
\end{array}\right.
  \end{eqnarray*}
We have  $\boldsymbol{\alpha}\cap A\neq \emptyset$ or $\boldsymbol{\beta}\cap B\neq \emptyset$. Note that the support of the vector $\boldsymbol{\alpha+\beta}$ is equal to $(\mbox{supp}(\boldsymbol{\alpha}) \cup \mbox{supp}(\boldsymbol{\beta}))\backslash (\mbox{supp}(\boldsymbol{\alpha}) \cap \mbox{supp}(\boldsymbol{\beta}))$. From  $\boldsymbol{\alpha}\cap A\neq \emptyset$ and $\boldsymbol{\beta}\cap A=\emptyset$, we have $(\boldsymbol{\alpha+\beta})\cap A\neq \emptyset$, which is a contradiction with $(\boldsymbol{\alpha+\beta})\cap A=\emptyset.$ Similarly, we derive  
$(\boldsymbol{\alpha+\beta})\cap B\neq\emptyset$ from $\boldsymbol{\beta}\cap B\neq \emptyset$. Hence there is no $\mathbf{a}=\boldsymbol{\alpha}+w\boldsymbol{\beta}$ such that $T=2$.

(3) $T=1$. In this case we have $\mbox{wt}(c_{D^{*}}(\mathbf{a}))=2^{|A|+|B|-1}$.
If $\boldsymbol{\alpha}\in \mathcal{U}_1 $ in Lemma \ref{lem4.1},  then $T=1$ if and only if $\boldsymbol{\beta}\in \mathcal{U}_2 \cup \mathcal{U}_3$.  If $\boldsymbol{\alpha}\in \mathcal{U}_2 $ and $\boldsymbol{\beta}\in \mathcal{U}_1\cup \mathcal{U}_3 $, then $T=1$. If $\boldsymbol{\alpha}\in \mathcal{U}_2 $ and $\boldsymbol{\beta}\in \mathcal{U}_2 $, then $T= \psi(\boldsymbol{\alpha}+\boldsymbol{\beta}|B)=1$ if and only if $\boldsymbol{\alpha}+\boldsymbol{\beta}\cap B=\emptyset$, which is equivalent to $\mbox{supp}(\boldsymbol{\alpha})\cap B= \mbox{supp}(\boldsymbol{\beta})\cap B$. If $\boldsymbol{\alpha}\in \mathcal{U}_3 $ and $\boldsymbol{\beta}\in \mathcal{U}_1 $, then $T=1$. If $\boldsymbol{\alpha}\in \mathcal{U}_3 $ and $\boldsymbol{\beta}\in \mathcal{U}_3\cup \mathcal{U}_4 $, then $T=1$ if and only if $\mbox{supp}(\boldsymbol{\alpha})\cap A= \mbox{supp}(\boldsymbol{\beta})\cap A$.  If $\boldsymbol{\alpha}\in \mathcal{U}_4 $ and $\boldsymbol{\beta}\in \mathcal{U}_2 $, then $T=1$ if and only if $\mbox{supp}(\boldsymbol{\alpha})\cap B= \mbox{supp}(\boldsymbol{\beta})\cap B$. 
By Lemma \ref{lem4.1}, the number of such $\mathbf{a}=\boldsymbol{\alpha}+w\boldsymbol{\beta}$ is $$ 3\times 4^{m-|A\cup B|}(2^{|A\backslash B|+|B\backslash A| }-1). $$

(4) $T=0$. In this case we have $\mbox{wt}(c_{D^{*}}(\mathbf{a}))=3\times2^{|A|+|B|-2}$. By Lemma \ref{lem4.1}, the number of such $\mathbf{a}=\boldsymbol{\alpha}+w\boldsymbol{\beta}$ is  $4^{m-|A\cup B|}[4^{|A\cup B|}-1-3(2^{|A\backslash B|+|B\backslash A|}-1)]$.

This completes the proof.
\end{proof}

\begin{cor} {\rm Let $A$ be a  subset of $[m]$ and $D=\Bbb F_2^m+w\Delta_A\subset \Bbb F_4^m$ or $D=\Delta_A+w\Bbb F_2^m\subset \Bbb F_4^m$. Then $\mathcal{C}_{D^{*}}$ is a $[2^{m+|A|}-1, m,2^{m+|A|-1}]$ two-weight quaternary code and its
 weight distribution  is presented in Table 3.  \begin{table}[h]
\caption{Weight distribution of the code in Corollary 4.3}
\begin{center}
\begin{tabu} to 0.6\textwidth{X[1,c]|X[2,c]}
\hline
Weight &Frequency\\
\hline
$0$&$1$\\
\hline
$2^{m+|A|-1}$&$3\times(2^{m-| A|}-1)$\\
 \hline
$3\times2^{m+|A|-2}$&$2^{2m}-1-3\times(2^{m-| A|}-1)$\\
\hline
\end{tabu}
\end{center}
\end{table}

\begin{thm}\label{thm4.4} {\rm Let $A,B$ be  two  subsets of $[m]$ and $D=\Delta_A+w\Delta_B\subset \Bbb F_4^m$. Then $\mathcal{C}_{D^{c}}$ in \eqref{eq2} is a $[4^{m}-2^{|A|+|B|}, m,3\times 2^{2m-2}-3\times2^{|A|+|B|-2}]$ quaternary code and its weight distribution
 is presented in Table 4.  Moreover, the code $\mathcal{C}_{D^{c}}$ is a  Griesmer code.
  \begin{table}[h]
\caption{Weight distribution of the code in Theorem \ref{thm4.4}}
\begin{center}
\begin{tabu} to 0.9\textwidth{X[1,c]|X[1.7,c]}
\hline
Weight &Frequency\\
\hline
$0$&$1$\\
\hline
$3\times 2^{2m-2}-3\times2^{|A|+|B|-2}$&$4^{m-|A\cup B|}[4^{|A\cup B|}-1-3(2^{|A\backslash B|+|B\backslash A|}-1)]$\\
 \hline
$3\times 2^{2m-2}-2^{|A|+|B|-1}$&$4^{m-|A\cup B|}\times3(2^{|A\backslash B|+|B\backslash A|}-1)$\\
\hline

 $3\times 2^{2m-2}$&$4^{m-|A\cup B|}-1$\\
 \hline

\end{tabu}
\end{center}
\end{table}

}

\end{thm}
\begin{proof} By Equation \eqref{eq4.5}, we have the weight distribution of the code.

By the Griesmer bound, we have \begin{eqnarray*}&&\sum_{i=0}^{m-1}\bigg\lceil {\frac{3(2^{2m-2}-2^{|A|+|B|-2})}{4^i}} \bigg\rceil\\
&&=\sum_{i=0}^{m-1}{\frac{3\times 2^{2m-2}}{4^i}} -\sum_{i=0}^{m-1}\bigg\lfloor {\frac{3\times2^{|A|+|B|-2}}{4^i}} \bigg\rfloor \\
&&=3\times 2^{2m-2}+3\times 2^{2m-4}+\cdots+3\\
&&-(3\times2^{|A|+|B|-2}+3\times2^{|A|+|B|-4}+\cdots+X+Y),
 \end{eqnarray*}
where $X=3$ and $Y=0$ if $|A|+|B|-2$ is even; and  $X=6$ and $Y=1$ if $|A|+|B|-2$ is odd.  Then
\begin{eqnarray*}&&\sum_{i=0}^{m-1}\bigg\lceil {\frac{3(2^{2m-2}-2^{|A|+|B|-2})}{4^i}} \bigg\rceil\\
&&=\frac{3\times 2^{2m-2}-3\times\frac14}{1-\frac 14}-\frac{3\times 2^{|A|+|B|-2} -X\times \frac 14}{1-\frac 14}-Y\\
&&=2^{2m}-1-(2^{|A|+|B|}-1)=4^{m}-2^{|A|+|B|}.
 \end{eqnarray*}
 
 This completes the proof.
\end{proof}

The following are examples of Proposition 4.2 and Theorem 4.4. 

\begin{example} {\rm Let $m=4$, $A=\{1,2,3\}$, and $B=\{3,4\}$.

(1) The code $\mathcal{C}_{D^{*}}$ in Proposition 4.2  is a two-weight quaternary $[31, 4,16]$ linear code with weight enumerator $1+21z^{16}+234z^{24}$.   In fact, the optimal quaternary linear code has parameter $[ 31, 4, 22] $, according to \cite{G2}.

(2) The code $\mathcal{C}_{D^{c}}$ in Theorem 4.4 is a two-weight quaternary $[224, 4, 168]$ linear code with weight enumerator  $1+21z^{168}+234z^{176}$. According to \cite{G2}, the code is optimal.

}
\end{example}

\begin{example} {\rm Let $m=4$, $A=\{2,3\}$, and $B=\{3,4\}$. 

(1) The code $\mathcal{C}_{D^{*}}$ in Proposition 4.2  is a two-weight quaternary $[15, 3,8]$ linear code with weight enumerator $1+9z^{8}+54z^{12}$.  In fact, the optimal quaternary linear code has parameter $[ 15, 3, 11] $, according to \cite{G2}.

(2) The code $\mathcal{C}_{D^{c}}$ in Theorem 4.4 is a three-weight quaternary $[240, 4, 180]$ linear code with weight enumerator  $1+216z^{180}+36z^{184}+3z^{192}$. According to \cite{G2}, the code is optimal.

}
\end{example}

Next we consider the case of a simplicial  complex with two maximal elements.

\begin{prop}\label{prop4.7}{\rm Let $\Delta $ be a simplicial  complex with two maximal elements $A, B\subseteq [m]$.  Let $D=\Delta+w\Delta \subset \Bbb F_4^m$.  Then $\mathcal{C}_{D^{*}}$ in \eqref{eq2} is a $[(2^{|A|}+2^{|B|}-2^{|A\cap B|})^{2}-1, |A\cup B|]$ quaternary code and its weight distribution
 is presented in Table 5.  \begin{table}[h]
\caption{Weight distribution of the code in Proposition \ref{prop4.7}}
\begin{center}
\begin{tabu} to 1\textwidth{X[2,c]|X[1,c]}
\hline
Weight &Frequency\\
\hline
$0$&$1$\\
\hline
$2^{|A|}(3\times 2^{|A|-2}+2^{|B|}-2^{|A\cap B|})$&$3(2^{|A\backslash B|}-1)$\\
\hline
$2^{|B|}( 2^{|A|}+3\times2^{|B|-2}-2^{|A\cap B|})$&$3(2^{|B\backslash A|}-1)$\\
\hline
$(2^{|A|}+2^{|B|})(3\times2^{|A|-2}+3\times2^{|B|-2}- 2^{|A\cap B|})$&$3(2^{|A\backslash B|}-1)(2^{|B\backslash A|}-1)$\\
 \hline
 % $(2^{|A|}+2^{|B|}-2^{|A\cap B|+1})(3\times 2^{|A|-2}+3\times 2^{|B|-2}-2^{|A\cap B|-1})$&$(2^{|A\backslash B|}-1)(2^{|B\backslash A|}-1)$\\
% \hline
 $\frac342^{|A|}(2^{|A|}+2^{|B|+1}-2^{|A\cap B|+1})$&$(2^{|A\backslash B|}-1)(2^{|A\backslash B|}-2)$\\
 \hline
  $\frac342^{|B|}(2^{|A|+1}+2^{|B|}-2^{|A\cap B|+1})$&$(2^{|B\backslash A|}-1)(2^{|B\backslash A|}-2)$\\
 \hline

  $\frac34(2^{|A|}+2^{|B|}-2^{|A\cap B|})^{2}-\frac 14 (2^{|A|}-2^{|A\cap B|})(2^{|B|}-2^{|A\cap B|})+\frac14 2^{|A\cap B|}(2^{|A|}+2^{|B|}-2^{|A\cap B|+1})$&$6(2^{|A\backslash B|}-1)(2^{|B\backslash A|}-1)$\\
 \hline
   $\frac34(2^{|A|}+2^{|B|}-2^{|A\cap B|})^{2}+\frac14 2^{|A\cap B|}(2^{|A|+1}-3\times 2^{|A\cap B|})$&$3(2^{|A\backslash B|}-1)(2^{|B\backslash A|}-1)(2^{|B\backslash A|}-2)$\\
 \hline
   $\frac34(2^{|A|}+2^{|B|}-2^{|A\cap B|})^{2}+\frac14 2^{|A\cap B|}(2^{|B|+1}-3\times 2^{|A\cap B|})$&$3(2^{|A\backslash B|}-1)(2^{|B\backslash A|}-1)(2^{|A\backslash B|}-2)$\\
 \hline

  $\frac 34 (2^{|A|}+2^{|B|})( 2^{|A|}+2^{|B|}-2^{|A\cap B|+1})$&$(2^{|A\backslash B|}-1)(2^{|B\backslash A|}-1)(2^{|A\backslash B|}-2)(2^{|B\backslash A|}-2)$\\
 \hline
  $\frac34(2^{|A|}+2^{|B|}-2^{|A\cap B|})^{2}$&$4^{|A\cup B|}-4^{|A\backslash B|+|B\backslash A|}$\\
 \hline
\end{tabu}
\end{center}
\end{table}
}

\end{prop}

\begin{proof} It is easy to check that the length of the  code $\mathcal{C}_{D^{*}}$ is $|D^{*}|=(2^{|A|}+2^{|B|}-2^{|A\cap B|})^{2}-1$.  

By Lemma \ref{lem2.1} 
\begin{multline}\label{eq4.7}
\chi_{\mathbf {u}}(\Delta)=\sum_{\mathbf x\in \Delta} (-1)^{\mathbf {u}\cdot\bf x}=\mathcal{H}_{\Delta}((-1)^{u_1},(-1)^{u_2},\ldots, (-1)^{u_m})\\=\prod_{i\in A}(1+(-1)^{u_i})+\prod_{i\in B}(1+(-1)^{u_i})-\prod_{i\in A\cap B}(1+(-1)^{u_i})\\
=2^{|A|}\psi(\mathbf {u}|A)+2^{|B|}\psi(\mathbf {u}|B)-2^{|A\cap B|}\psi(\mathbf {u}|A\cap B).
\end{multline}
Suppose that $\mathbf{a}=\boldsymbol{\alpha}+w\boldsymbol{\beta}$.
%$\mathbf{d}={\bf d_1}+w{\bf d_2}$, where ${\boldsymbol{\alpha}=(\alpha_1, \cdots, \alpha_m),\boldsymbol{\beta}=(\beta_1, \cdots, \beta_m)}\in \Bbb F_2^m$, ${\bf d_1}\in D_{1}$, and ${\bf d_2}\in D_{2}$.  
By Equations   \eqref{eq4} and \eqref{eq4.7}
\begin{eqnarray*} \label{eq7}
&&\mbox{wt}(c_{D^{*}}(\mathbf{a}))=\mbox{wt}(c_{D^{}}(\mathbf{a}))= \frac 3 4|D^{}|\nonumber\\
&&-\frac 14[ \chi_{\boldsymbol{\alpha} }(\Delta)\chi_{\boldsymbol{\beta} }(\Delta) + \chi_{\boldsymbol{\beta} }(\Delta)\chi_{\boldsymbol{\alpha+\beta} }(\Delta)+ \chi_{\boldsymbol{\alpha+\beta} }(\Delta)\chi_{\boldsymbol{\alpha} }(\Delta)].\end{eqnarray*}
By Lemma \ref{lem4.1}, \begin{eqnarray*} \chi_{\boldsymbol{\alpha} }(\Delta)=\left\{\begin{array}{llll}
2^{|A|}+2^{|B|}-2^{|A\cap B|}, \mbox{ if } \boldsymbol{\alpha}\in \mathcal{U}_1,\\

2^{|A|}-2^{|A\cap B|}, \mbox{ if } \boldsymbol{\alpha}\in \mathcal{U}_2,\\

2^{|B|}-2^{|A\cap B|}, \mbox{ if } \boldsymbol{\alpha}\in \mathcal{U}_3,\\

-2^{|A\cap B|}, \mbox{ if } \boldsymbol{\alpha}\in \mathcal{U}_4,\\
0, \mbox{ if } \boldsymbol{\alpha}\in \mathcal{U}_5.\\

\end{array}\right.  \end{eqnarray*}

Due to the above value distribution, we determine the location of $ \boldsymbol{\alpha}+ \boldsymbol{\beta}$ in the following table. 

\begin{table}[!htbp]
\centering
\begin{tabular}{|c|c|c|c|c|c|c|c}
\hline
\diagbox{ $\boldsymbol{\alpha}$}{$ \boldsymbol{\alpha}+ \boldsymbol{\beta}$}{$ \boldsymbol{\beta}$}&$\mathcal{U}_1$&$\mathcal{U}_2$&$\mathcal{U}_3$&$\mathcal{U}_4$&$\mathcal{U}_5$\\ %添加斜线表头
\hline
$\mathcal{U}_1$&$\mathcal{U}_1$&$\mathcal{U}_2$&$\mathcal{U}_3$&$\mathcal{U}_4$&$\mathcal{U}_5$\\
\hline
$\mathcal{U}_2$&$\mathcal{U}_2$&$\mathcal{U}_1$\mbox{ or } $\mathcal{U}_2$&$\mathcal{U}_4$&$\mathcal{U}_3$\mbox{ or } $\mathcal{U}_4$&$\mathcal{U}_5$\\
\hline
$\mathcal{U}_3$&$\mathcal{U}_3$&$\mathcal{U}_4$&$\mathcal{U}_1$\mbox{ or }$\mathcal{U}_3$&$\mathcal{U}_2$\mbox{ or } $\mathcal{U}_4$&$\mathcal{U}_5$\\
\hline
$\mathcal{U}_4$&$\mathcal{U}_4$&$\mathcal{U}_3$\mbox{ or } $\mathcal{U}_4$&$\mathcal{U}_2$\mbox{ or }$\mathcal{U}_4$&$\mathcal{U}_1$\mbox{ or }$\mathcal{U}_2$\mbox{ or } $\mathcal{U}_3$\mbox{ or }$\mathcal{U}_4$&$\mathcal{U}_5$\\
\hline
$\mathcal{U}_5$&$\mathcal{U}_5$&$\mathcal{U}_5$&$\mathcal{U}_5$&$\mathcal{U}_5$&\\
\hline
\end{tabular}
\end{table}
It needs to be further pointed out that if $ \boldsymbol{\alpha}, \boldsymbol{\beta}\in \mathcal{U}_4$, then \begin{eqnarray*}  \boldsymbol{\alpha}+ \boldsymbol{\beta}\in\left\{\begin{array}{llll}
\mathcal{U}_1, \mbox{ if } \mbox{supp}(\boldsymbol{\alpha})\cap A=\mbox{supp}(\boldsymbol{\beta})\cap A, \mbox{supp}(\boldsymbol{\alpha})\cap B=\mbox{supp}(\boldsymbol{\beta})\cap B,\\

\mathcal{U}_2, \mbox{ if } \mbox{supp}(\boldsymbol{\alpha})\cap A=\mbox{supp}(\boldsymbol{\beta})\cap A, \mbox{supp}(\boldsymbol{\alpha})\cap B\neq\mbox{supp}(\boldsymbol{\beta})\cap B,\\

\mathcal{U}_3, \mbox{ if } \mbox{supp}(\boldsymbol{\alpha})\cap A\neq\mbox{supp}(\boldsymbol{\beta})\cap A, \mbox{supp}(\boldsymbol{\alpha})\cap B=\mbox{supp}(\boldsymbol{\beta})\cap B,\\

\mathcal{U}_4, \mbox{ if } \mbox{supp}(\boldsymbol{\alpha})\cap A\neq\mbox{supp}(\boldsymbol{\beta})\cap A, \mbox{supp}(\boldsymbol{\alpha})\cap B\neq\mbox{supp}(\boldsymbol{\beta})\cap B.\\

\end{array}\right.  \end{eqnarray*}

Let $T=\chi_{\boldsymbol{\alpha} }(\Delta)\chi_{\boldsymbol{\beta} }(\Delta) + \chi_{\boldsymbol{\beta} }(\Delta)\chi_{\boldsymbol{\alpha+\beta} }(\Delta)+ \chi_{\boldsymbol{\alpha+\beta} }(\Delta)\chi_{\boldsymbol{\alpha} }(\Delta).$  Suppose that $\boldsymbol{\alpha}\in \mathcal{U}_1$.  If  $\boldsymbol{\beta}\in \mathcal{U}_1$,  then $T=3(2^{|A|}+2^{|B|}-2^{|A\cap B|})^{2}.$ If $\boldsymbol{\beta}\in \mathcal{U}_2$, Then $T=2(2^{|A|}+2^{|B|}-2^{|A\cap B|})(2^{|A|}-2^{|A\cap B|})+(2^{|A|}-2^{|A\cap B|})^{2}. $ If $\boldsymbol{\beta}\in \mathcal{U}_3$, Then $T=2(2^{|A|}+2^{|B|}-2^{|A\cap B|})(2^{|B|}-2^{|A\cap B|})+(2^{|B|}-2^{|A\cap B|})^{2}. $ If $\boldsymbol{\beta}\in \mathcal{U}_4$, Then $T=2(2^{|A|}+2^{|B|}-2^{|A\cap B|})(-2^{|A\cap B|})+(2^{|A\cap B|})^{2}. $ If $\boldsymbol{\beta}\in \mathcal{U}_5$, Then $T=0. $

Similarly, we can determine the values of $T$ under the condition of $\boldsymbol{\alpha}\in \mathcal{U}_i, i=2,3,4,5$. Then the results follow from  Lemma 4.1. 
\end{proof}

%\begin{cor} {\rm Suppose that $A=B$ in Lemma 4.1. Then the code $\mathcal{C}_{D^{*}}$ is a $[4^{|A|}-1, |A|, 3\times 4^{|A|-1}]$ Griesmer quaternary code.

%}
%\end{cor}

\begin{rem}{\rm

By massive computation, weight distributions of quaternary codes  can be also determined in the case of $D=D_{1}+wD_{2}$, where $D_{1}$ is generated by two maximal elements $A, B\subseteq [m]$ and  $D_{2}$ is generated by two maximal elements $C, F\subseteq [m].$ 
}
\end{rem}

{\color{blue}For some special sets $A, B$, we obtain some few-weight quaternary codes. %the number of weight of the code in Theorem 4.7 can be further reduce     as follows: 

\begin{cor} {\rm Let $\Delta $ be a simplicial  complex with two maximal elements $A, B\subseteq [m]$.  Let $D=\Delta+w\Delta \subset \Bbb F_4^m$. 

(i) If $A\cap B=\emptyset$ and $|A|=|B|=1$, then $\mathcal{C}_{D^{*}}$ in  Proposition  4.7 is a three-weight $[8,2,5]$ quaternary code and its weight enumerator $1+6z^{5}+6z^{7}+3z^{8}$.

(ii) If $A\cap B=\emptyset$ and $|A|=|B|>1$, then $\mathcal{C}_{D^{*}}$ in Proposition 4.7 is a six-weight $[(2^{|A|+1}-1)^{2}-1, 2|A|]$ quaternary code and its weight distribution
 is presented in Table 6.  \begin{table}[h]
\caption{Weight distribution of the code in Corollary 4.9 (ii)}
\begin{center}
\begin{tabu} to 0.8\textwidth{X[2,c]|X[2,c]}
\hline
Weight &Frequency\\
\hline
$0$&$1$\\
\hline
$2^{|A|}(3\times2^{|A|-2}+2^{|A|}-1)$&$6(2^{| A|}-1)$\\
 \hline
$2^{|A|+1}(3\times2^{|A|-1}-1)$&$3(2^{| A|}-1)^{2}$\\
 \hline
 
 $3\times2^{|A|-1}(3\times2^{|A|-1}-1)$&$3(2^{| A|}-1)(2^{| A|}-2)$\\
 \hline
 
  $11\times2^{2|A|-2}-2^{|A|+1}$&$6(2^{| A|}-1)^{2}$\\
 \hline 
 
  $3\times2^{2|A|}-3\times2^{|A|}+2^{|A|-1}$&$6(2^{| A|}-1)^{2}(2^{| A|}-2)$\\
 \hline
 
  $3\times2^{|A|}(2^{|A|}-1)$&$(2^{| A|}-1)^{2}(2^{| A|}-2)^{2}$\\
 \hline
\end{tabu}
\end{center}
\end{table}

(iii) If $A\cap B\neq\emptyset$, $|A|=|B|$,  and $|A\backslash B|=|B\backslash A|=1$, then $\mathcal{C}_{D^{*}}$ in Proposition 4.7 is a four-weight $[9\times 2^{2|A|-2}-1, |A|+1]$ quaternary code and its weight distribution
 is presented in Table 7.  \begin{table}[h]
\caption{Weight distribution of the code in Corollary 4.9 (iii)}
\begin{center}
\begin{tabu} to 0.6\textwidth{X[2,c]|X[1,c]}
\hline
Weight &Frequency\\
\hline
$0$&$1$\\
\hline
$2^{2|A|-2}+2^{2|A|}$&$6$\\
\hline

$2^{2|A|+1}$&$3$\\
 \hline
 % $(2^{|A|}+2^{|B|}-2^{|A\cap B|+1})(3\times 2^{|A|-2}+3\times 2^{|B|-2}-2^{|A\cap B|-1})$&$(2^{|A\backslash B|}-1)(2^{|B\backslash A|}-1)$\\
% \hline

  $26\times 2^{2|A|-4}$&$6$\\
 \hline
  $27\times 2^{2|A|-4}$&$4^{|A|+1}-16$\\
 \hline
\end{tabu}
\end{center}
\end{table}

}
\end{cor}

}

Similar to Theorem 4.4, we have the following theorem.

\begin{thm}\label{thm4.6}{\rm Let $\Delta $ be a simplicial  complex with two maximal elements $A, B\subseteq [m]$.  Let $D=\Delta+w\Delta \subset \Bbb F_4^m$.  Then $\mathcal{C}_{D^{c}}$ in \eqref{eq2} is a $[4^{m}-(2^{|A|}+2^{|B|}-2^{|A\cap B|})^{2}, m]$ quaternary code and its weight distribution
 is presented in Table 8.  \begin{table}[h]
\caption{Weight distribution of the code in Theorem \ref{thm4.6}}
\begin{center}
\begin{tabu} to 1\textwidth{X[2,c]|X[1.5,c]}
\hline
Weight &Frequency\\
\hline
$0$&$1$\\
\hline
$3\times 2^{2m-2}-2^{|A|}(3\times 2^{|A|-2}+2^{|B|}-2^{|A\cap B|})$&$3(2^{|A\backslash B|}-1)4^{m-|A\cup B|}$\\
\hline
$3\times 2^{2m-2}-2^{|B|}( 2^{|A|}+3\times2^{|B|-2}-2^{|A\cap B|})$&$3(2^{|B\backslash A|}-1)4^{m-|A\cup B|}$\\
\hline
$3\times 2^{2m-2}-(2^{|A|}+2^{|B|})(3\times2^{|A|-2}+3\times2^{|B|-2}- 2^{|A\cap B|})$&$3(2^{|A\backslash B|}-1)(2^{|B\backslash A|}-1)4^{m-|A\cup B|}$\\
 \hline

 $3\times 2^{2m-2}-\frac342^{|A|}(2^{|A|}+2^{|B|+1}-2^{|A\cap B|+1})$&$(2^{|A\backslash B|}-1)(2^{|A\backslash B|}-2)4^{m-|A\cup B|}$\\
 \hline
  $3\times 2^{2m-2}-\frac342^{|B|}(2^{|A|+1}+2^{|B|}-2^{|A\cap B|+1})$&$(2^{|B\backslash A|}-1)(2^{|B\backslash A|}-2)4^{m-|A\cup B|}$\\
 \hline

  $3\times 2^{2m-2}-\frac34(2^{|A|}+2^{|B|}-2^{|A\cap B|})^{2}+\frac 14 (2^{|A|}-2^{|A\cap B|})(2^{|B|}-2^{|A\cap B|})-\frac14 2^{|A\cap B|}(2^{|A|}+2^{|B|}-2^{|A\cap B|+1})$&$6(2^{|A\backslash B|}-1)(2^{|B\backslash A|}-1)4^{m-|A\cup B|}$\\
 \hline
   $3\times 2^{2m-2}-\frac34(2^{|A|}+2^{|B|}-2^{|A\cap B|})^{2}-\frac14 2^{|A\cap B|}(2^{|A|+1}-3\times 2^{|A\cap B|})$&$3(2^{|A\backslash B|}-1)(2^{|B\backslash A|}-1)(2^{|B\backslash A|}-2)4^{m-|A\cup B|}$\\
 \hline
   $3\times 2^{2m-2}-\frac34(2^{|A|}+2^{|B|}-2^{|A\cap B|})^{2}-\frac14 2^{|A\cap B|}(2^{|B|+1}-3\times 2^{|A\cap B|})$&$3(2^{|A\backslash B|}-1)(2^{|B\backslash A|}-1)(2^{|A\backslash B|}-2)4^{m-|A\cup B|}$\\
 \hline

  $3\times 2^{2m-2}-\frac 34 (2^{|A|}+2^{|B|})( 2^{|A|}+2^{|B|}-2^{|A\cap B|+1})$&$(2^{|A\backslash B|}-1)(2^{|B\backslash A|}-1)(2^{|A\backslash B|}-2)(2^{|B\backslash A|}-2)4^{m-|A\cup B|}$\\
 \hline
  $3\times 2^{2m-2}-\frac34(2^{|A|}+2^{|B|}-2^{|A\cap B|})^{2}$&$(4^{|A\cup B|}-4^{|A\backslash B|+|B\backslash A|})4^{m-|A\cup B|}$\\
 \hline
  $3\times 2^{2m-2}$&$4^{m-|A\cup B|}-1$\\
 \hline
\end{tabu}
\end{center}
\end{table}
}

\end{thm}

The following is an example of Proposition 4.7 and Theorem 4.10. 

\begin{example} {\rm Let $m=4$, $A=\{1,2,3\}$, and $B=\{3,4\}$.

(1) The code $\mathcal{C}_{D^{*}}$ in Proposition 4.7  is a seven-weight quaternary $[99, 4,36]$ linear code with weight enumerator $$1+3z^{36}+9z^{64}+6z^{72}+192z^{75}+18z^{76}+9z^{84}+18z^{88}.$$   In fact, the optimal quaternary linear code has parameter $[99, 4,73] $ and its dual code has parameters $[99, 95,2]$, according to \cite{G2}.

(2) The code $\mathcal{C}_{D^{c}}$ in Theorem 4.10 is a seven-weight quaternary $[156, 4, 104]$ linear code with weight enumerator $$1+18z^{104}+9z^{108}+18z^{116}+192z^{117}+6z^{120}+9z^{128}+3z^{156}.$$  In fact, the optimal quaternary linear code has parameter $[156, 4, 116]$ and its dual code has parameters $[156, 152,2]$, according to \cite{G2}.

}
\end{example}

{\color{blue} By Corollary  4.9 and Theorem  4.10, we obtain some few-weight quaternary  codes.

\begin{cor} {\rm Let $\Delta $ be a simplicial  complex with two maximal elements $A, B\subseteq [m]$.  Let $D=\Delta+w\Delta \subset \Bbb F_4^m$. 

(i) If $A\cap B=\emptyset$ and $|A|=|B|=1$, then $\mathcal{C}_{D^{c}}$ in Theorem 4.10 is a four-weight $[4^{m}-9,m]$ quaternary code and its weight enumerator $$1+6\times4^{m-2}z^{3\times 2^{2m-2}-5}+6\times4^{m-2}z^{3\times 2^{2m-2}-7}+3\times4^{m-2}z^{3\times 2^{2m-2}-8}+(4^{m-2}-1)z^{3\times 2^{2m-2}}.$$

(ii) If $A\cap B=\emptyset$ and $|A|=|B|>1$, then $\mathcal{C}_{D^{c}}$ in Theorem 4.10 is a seven-weight $[4^{m}-(2^{|A|+1}-1)^{2}, m]$ quaternary code and its weight distribution
 is presented in Table 9.  \begin{table}[h]
\caption{Weight distribution of the code in Corollary 4.12 (ii)}
\begin{center}
\begin{tabu} to 1\textwidth{X[2,c]|X[1.8,c]}
\hline
Weight &Frequency\\
\hline
$0$&$1$\\
\hline
$3\times 2^{2m-2}-2^{|A|}(3\times2^{|A|-2}+2^{|A|}-1)$&$6(2^{| A|}-1)4^{m-2|A|}$\\
 \hline
$3\times 2^{2m-2}-2^{|A|+1}(3\times2^{|A|-1}-1)$&$3(2^{| A|}-1)^{2}4^{m-2|A|}$\\
 \hline
 
 $3\times 2^{2m-2}-3\times2^{|A|-1}(3\times2^{|A|-1}-1)$&$3(2^{| A|}-1)(2^{| A|}-2)4^{m-2|A|}$\\
 \hline
 
  $3\times 2^{2m-2}-11\times2^{2|A|-2}-2^{|A|+1}$&$6(2^{| A|}-1)^{2}4^{m-2|A|}$\\
 \hline 
 
  $3\times 2^{2m-2}-3\times2^{2|A|}-3\times2^{|A|}+2^{|A|-1}$&$6(2^{| A|}-1)^{2}(2^{| A|}-2)4^{m-2|A|}$\\
 \hline
 
  $3\times 2^{2m-2}-3\times2^{|A|}(2^{|A|}-1)$&$(2^{| A|}-1)^{2}(2^{| A|}-2)^{2}4^{m-2|A|}$\\
 \hline
   $3\times 2^{2m-2}$&$4^{m-2|A|}-1$\\
 \hline
\end{tabu}
\end{center}
\end{table}

(iii) If $A\cap B\neq\emptyset$, $|A|=|B|$,  and $|A\backslash B|=|B\backslash A|=1$, then $\mathcal{C}_{D^{c}}$ in Theorem 4.10 is a five-weight $[4^{m}-9\times 2^{2|A|-2}, m]$ quaternary code and its weight distribution
 is presented in Table 10.  \begin{table}[h]
\caption{Weight distribution of the code in Corollary 4.12 (iii)}
\begin{center}
\begin{tabu} to 0.8\textwidth{X[2,c]|X[2,c]}
\hline
Weight &Frequency\\
\hline
$0$&$1$\\
\hline
$3\times 2^{2m-2}-2^{2|A|-2}+2^{2|A|}$&$6\times 4^{m-|A|-1}$\\
\hline

$3\times 2^{2m-2}-2^{2|A|+1}$&$3\times 4^{m-|A|-1}$\\
 \hline
 % $(2^{|A|}+2^{|B|}-2^{|A\cap B|+1})(3\times 2^{|A|-2}+3\times 2^{|B|-2}-2^{|A\cap B|-1})$&$(2^{|A\backslash B|}-1)(2^{|B\backslash A|}-1)$\\
% \hline

  $3\times 2^{2m-2}-26\times 2^{2|A|-4}$&$6\times4^{m-|A|-1}$\\
 \hline
  $3\times 2^{2m-2}-27\times 2^{2|A|-4}$&$(4^{|A|+1}-16)4^{m-|A|-1}$\\
 \hline
   $3\times 2^{2m-2}$&$4^{m-|A|-1}-1$\\
 \hline
\end{tabu}
\end{center}
\end{table}

}
\end{cor}
}

}

\end{cor}

\section{weight distributions of the subfield codes}

In this section, we will determine weight distributions of the binary subfield codes of those quaternary codes obtained in Section 4.

\begin{prop}{\rm Let $A,B$ be  two  subsets of $[m]$ and $D=\Delta_A+w\Delta_B\subset \Bbb F_4^m$. Then the subfield code $\mathcal{C}^{(2)}_{D^{*}}$ with respect to  $\mathcal{C}_{D^{*}}$ in Proposition 4.2  is a $[2^{|A|+|B|}-1, |A|+|B|,2^{|A|+|B|-1} ]$ one-weight binary linear code.

}
\end{prop}

\begin{proof} By Theorem 3.2, $\mathcal{C}^{(2)}_{D^{*}}$ can be generated by 
$$\mathcal{C}^{(2)}_{D^{*}}=\{c_{D^{}}(\boldsymbol{\alpha}, \boldsymbol{\beta})=(\boldsymbol{\alpha}\cdot \mathbf{d}_{2}+\boldsymbol{\beta}\cdot (\mathbf{d}_{1}+\mathbf{d}_{2} ))_{\mathbf{d}_{1}\in \Delta_{A}, \mathbf{d}_{2}\in \Delta_{B}}:\boldsymbol{\alpha}, \boldsymbol{\beta}\in \Bbb F_2^m\}.$$ Hence
\begin{eqnarray*} 
\mbox{wt}(c_{D^{*}}(\boldsymbol{\alpha}, \boldsymbol{\beta}))=\mbox{wt}(c_{D}(\boldsymbol{\alpha}, \boldsymbol{\beta}))
&=&|D|-\sum_{\mathbf d_1\in  \Delta_{A}}\sum_{\mathbf d_2\in \Delta_{B}}(\frac12 \sum_{y\in\Bbb F_2} (-1)^{(\boldsymbol{\alpha}{\mathbf d_2}+\boldsymbol{\beta}{( \mathbf d_{1}+ \mathbf d_2}))y})\\
&=&|D|-\frac12\sum_{\mathbf d_1\in  \Delta_{A}}\sum_{\mathbf d_2\in \Delta_{B}}(1+(-1)^{\boldsymbol{\alpha}{\mathbf d_2}+\boldsymbol{\beta}{ (\mathbf d_{1}+\mathbf d_2})})\\
&=&\frac 1 2|D|-\frac 12(\sum_{\mathbf d_1\in \Delta_{A}} (-1)^{\boldsymbol{\beta}\mathbf d_1})(\sum_{\mathbf d_2\in \Delta_{B}}(-1)^{(\boldsymbol{\alpha}+\boldsymbol{\beta}){\mathbf d_2}})\\
&=&2^{|A|+|B|-1}(1-\psi(\boldsymbol{\beta}|A)\psi(\boldsymbol{\alpha}+\boldsymbol{\beta}|B)). 
\end{eqnarray*} 

This completes the proof.
\end{proof}

\begin{thm}{\rm Let $A,B$ be  two  subsets of $[m]$ and $D=\Delta_A+w\Delta_B\subset \Bbb F_4^m$. Then the subfield code $\mathcal{C}^{(2)}_{D^{c}}$ with respect to $\mathcal{C}_{D^{c}}$ in Theorem 4.4  is a $[2^{2m}-2^{|A|+|B|}, 2m,2^{2m-1}-2^{|A|+|B|-1} ]$ two-weight binary linear code and its weight distribution is given by 
$$1+(4^{m}-2^{2m-|A|-|B|})z^{2^{2m-1}-2^{|A|+|B|-1}}+(2^{2m-|A|-|B|}-1)z^{2^{2m-1}}.$$
Moreover, the code $\mathcal{C}^{(2)}_{D^{c}}$ is a  Griesmer code.

}
\end{thm}

\begin{proof} Note that  $D^{c}=(\Delta_{A}^{c}+w\Bbb F_{2}^{m})\cup (\Delta_{A}+w\Delta_{B}^{c})$. By Theorem 3.2, $\mathcal{C}^{(2)}_{D^{c}}$ can be generated by 
$$\mathcal{C}^{(2)}_{D^{c}}=\{c_{D^{c}}(\boldsymbol{\alpha}, \boldsymbol{\beta}):\boldsymbol{\alpha}, \boldsymbol{\beta}\in \Bbb F_2^m\},$$ 
where $$c_{D^{c}}(\boldsymbol{\alpha}, \boldsymbol{\beta})=((\boldsymbol{\alpha}\cdot \mathbf d_{2}+\boldsymbol{\beta}\cdot (\mathbf d_{1}+\mathbf d_{2} )_{\mathbf d_{1}\in \Delta_{A}^{c}, \mathbf d_{2}\in \Bbb F_{2}^{m}}| (\boldsymbol{\alpha}\cdot \mathbf f_{2}+\boldsymbol{\beta}\cdot (\mathbf f_{1}+\mathbf f_{2} )_{\mathbf f_{1}\in \Delta_{A}, \mathbf f_{2}\in \Delta_{B}^{c}})).$$
Hence
\begin{eqnarray*} 
\mbox{wt}(c_{D^{c}}(\boldsymbol{\alpha}, \boldsymbol{\beta}))
&=&|D^{c}|-\sum_{\mathbf d_1\in  \Delta_{A}^{c}}\sum_{\mathbf d_2\in \Bbb F_{2}^{m}}\frac12 \sum_{y\in\Bbb F_2} (-1)^{(\boldsymbol{\alpha}{\mathbf d_2}+\boldsymbol{\beta}{( \mathbf d_{1}+  \mathbf d_2}))y}\\
&-&\sum_{\mathbf f_1\in  \Delta_{A}}\sum_{\mathbf f_2\in \Delta_{B}^{c}}\frac12 \sum_{z\in\Bbb F_2} (-1)^{(\boldsymbol{\alpha}{ \mathbf f_2}+\boldsymbol{\beta}{( \mathbf f_{1}+ \mathbf f_2}))z}\\
&=&\frac 1 2|D^{c}|-\frac 12(\sum_{\mathbf d_1\in \Delta_{A}^{c}} (-1)^{\boldsymbol{\beta} \mathbf d_1})(\sum_{\mathbf d_2\in \Bbb F_{2}^{m}}(-1)^{(\boldsymbol{\alpha}+\boldsymbol{\beta}){ \mathbf d_2}})\\
&-&\frac 12(\sum_{\mathbf f_1\in \Delta_{A}} (-1)^{\boldsymbol{\beta} \mathbf f_1})(\sum_{\mathbf f_2\in  \Delta_B^{c}}(-1)^{(\boldsymbol{\alpha}+\boldsymbol{\beta}){ \mathbf f_2}})\\
&=&2^{2m-1}(1-\delta_{\bf 0,\boldsymbol{\beta}}\delta_{\bf 0,\boldsymbol{\alpha+\beta}})-2^{|A|+|B|-1}(1-\psi(\boldsymbol{\beta}|A)\psi(\boldsymbol{\alpha}+\boldsymbol{\beta}|B)). 
\end{eqnarray*}  Then we have the weight distribution of the code by Proposition 5.1. 

By the Griesmer bound, we have \begin{eqnarray*}&&\sum_{i=0}^{2m-1}\bigg\lceil {\frac{2^{2m-1}-2^{|A|+|B|-1}}{2^i}} \bigg\rceil\\
&=&\sum_{i=0}^{2m-1}{\frac{2^{2m-1}}{2^i}} -\sum_{i=0}^{2m-1}\bigg\lfloor {\frac{2^{|A|+|B|-1}}{2^i}} \bigg\rfloor \\
&=&(2^{2m}-1)-(2^{|A|+|B|}-1)=2^{2m}-2^{|A|+|B|}.
 \end{eqnarray*}
 
This completes the proof.
\end{proof}

\begin{prop}{\rm

Let $\Delta $ be a simplicial  complex with two maximal elements $A, B\subseteq [m]$.  Let $D=\Delta+w\Delta \subset \Bbb F_4^m$.  Then the subfield code  $\mathcal{C}^{(2)}_{D^{*}}$ with respect to $\mathcal{C}_{D^{*}}$ in Proposition 4.7 is 
a $[(2^{|A|}+2^{|B|}-2^{|A\cap B|})^{2}-1,2|A\cup B|]$ binary code and its weight distribution
 is presented in Table 11.  \begin{table}[h]
\caption{Weight distribution of the code in Proposition 5.3}
\begin{center}
\begin{tabu} to 1\textwidth{X[2,c]|X[1,c]}
\hline
Weight &Frequency\\
\hline
$0$&$1$\\
\hline
$2^{|A|-1}( 2^{|A|}+2^{|B|}-2^{|A\cap B|})$&$2(2^{|A\backslash B|}-1)$\\
\hline
$2^{|B|-1}( 2^{|A|}+2^{|B|}-2^{|A\cap B|})$&$2(2^{|B\backslash A|}-1)$\\
\hline
$(2^{|A|-1}+2^{|B|-1})( 2^{|A|}+2^{|B|}-2^{|A\cap B|})$&$2(2^{|A\backslash B|}-1)(2^{|B\backslash A|}-1)$\\
 \hline
$2^{|B|-1}( 2^{|A|+1}+2^{|B|}-2^{|A\cap B|+1})$&$(2^{|B\backslash A|}-1)^{2}$\\
\hline
$2^{|A|-1}( 2^{|A|}+2^{|B|+1}-2^{|A\cap B|+1})$&$(2^{|A\backslash B|}-1)^{2}$\\
\hline

$(2^{|A|-1}+2^{|B|-1})( 2^{|A|}+2^{|B|}-2^{|A\cap B|+1})$&$(2^{|B\backslash A|}-1)^{2}(2^{|B\backslash A|}-1)^{2}$\\
\hline

   $\frac12(2^{|A|}+2^{|B|}-2^{|A\cap B|})^{2}-\frac12 (2^{|A|}-2^{|A\cap B|})(2^{|B|}-2^{|A\cap B|})$&$2(2^{|A\backslash B|}-1)(2^{|B\backslash A|}-1)$\\
 \hline
 
 $\frac12(2^{|A|}+2^{|B|}-2^{|A\cap B|})^{2}+\frac12 (2^{|A|}-2^{|A\cap B|})2^{|A\cap B|}$&$2(2^{|A\backslash B|}-1)(2^{|B\backslash A|}-1)^{2}$\\
 \hline
 
  $\frac12(2^{|A|}+2^{|B|}-2^{|A\cap B|})^{2}+\frac12 (2^{|B|}-2^{|A\cap B|})2^{|A\cap B|}$&$2(2^{|A\backslash B|}-1)^{2}(2^{|B\backslash A|}-1)$\\
 \hline
 
  $\frac12(2^{|A|}+2^{|B|}-2^{|A\cap B|})^{2}$&$4^{|A\cup B|}-4^{|A\backslash B|+|B\backslash A|}$\\
 \hline
\end{tabu}
\end{center}
\end{table}

}
\end{prop}

\begin{proof} By Theorem 3.2, $\mathcal{C}^{(2)}_{D^{*}}$ can be generated by 
$$\mathcal{C}^{(2)}_{D^{*}}=\{c_{D^{}}(\boldsymbol{\alpha}, \boldsymbol{\beta})=(\boldsymbol{\alpha}\cdot \mathbf d_{2}+\boldsymbol{\beta}\cdot (\mathbf d_{1}+\mathbf d_{2} ))_{\mathbf d_{1}, \mathbf d_{2}\in \Delta}:\boldsymbol{\alpha}, \boldsymbol{\beta}\in \Bbb F_2^m\}.$$Hence
\begin{eqnarray*} 
\mbox{wt}(c_{D}(\boldsymbol{\alpha}, \boldsymbol{\beta}))
&=&|D|-\sum_{\mathbf d_1\in  \Delta}\sum_{\mathbf d_2\in \Delta}(\frac12 \sum_{y\in\Bbb F_2} (-1)^{(\boldsymbol{\alpha}{ \mathbf d_2}+\boldsymbol{\beta}{( \mathbf d_{1}+ \mathbf d_2}))y})\\
&=&|D|-\frac12\sum_{\mathbf d_1\in  \Delta}\sum_{\mathbf d_2\in \Delta}(1+(-1)^{\boldsymbol{\alpha}{ \mathbf d_2}+\boldsymbol{\beta}{( \mathbf d_{1}+ \mathbf d_2})})\\
&=&\frac 1 2|D|-\frac 12(\sum_{\mathbf d_1\in \Delta} (-1)^{\boldsymbol{\beta} \mathbf d_1})(\sum_{\mathbf d_2\in \Delta}(-1)^{(\boldsymbol{\alpha}+\boldsymbol{\beta}){ \mathbf d_2}})\\
&=&\frac12 (2^{|A|}+2^{|B|}-2^{|A\cap B|})^{2}-\frac12 T_{\boldsymbol{\beta}}T_{\boldsymbol{\alpha}+\boldsymbol{\beta}},
\end{eqnarray*} 
where $T_{\mathbf {u}}=2^{|A|}\psi(\mathbf {u}|A)+2^{|B|}\psi(\mathbf {u}|B)-2^{|A\cap B|}\psi(\mathbf {u}|A\cap B)$.

By the proof of Proposition 4.7, we prove the results. 
\end{proof}

{\color{blue}
\begin{rem}{\rm It is noted that the subfield code $\mathcal{C}^{(2)}_{D^{*}}$ with respect to $\mathcal{C}_{D^{*}}$ in Theorem 5.3 has at  most  seven weights when the sets  $A,B$ have the same size.%s, then the subfield code  $\mathcal{C}^{(2)}_{D^{*}}$ with respect to $\mathcal{C}_{D^{*}}$ in Theorem 5.3 has at  most  seven weights.  
}
\end{rem}

}

\begin{thm}{\rm The dual code of the subfield code  $\mathcal{C}^{(2)}_{D^{*}}$ in Proposition 5.3 has minimum distance three and it is an almost optimal binary code with respect to the Sphere Packing Bound. }
\end{thm}
\begin{proof} By Theorem 3.2,  the subfield code  $\mathcal{C}^{(2)}_{D^{*}}$ in Proposition 5.3 has the following defining set: $$D^{(2)}=\{(\mathbf d_{2}, \mathbf d_{1}+\mathbf d_{2}): \mathbf d_{1}, \mathbf d_{2}\in \Delta \}=\{{\bf g}_{1}, {\bf g}_{2},\ldots, {\bf g}_{t}\},$$ where $t=(2^{|A|}+2^{|B|}-2^{|A\cap B|})^{2}$. 
Let $G$ be the $2m\times t$ matrix as follows:
\begin{equation*}G=[{\bf g}_1^T \; {\bf g}_2^T \; \cdots \; {\bf g}_t^T],\end{equation*}
where the column vector ${{\bf g}_i^T}$ denotes the transpose of a row vector ${\bf g}_i$. Let ${\bf e}_k=(e_1, e_2, \ldots, e_m)\in \Bbb F_2^m$, where
$e_k=1$ and $e_l=0 $ if $l\neq k$. Suppose that $i,j\in A\cup B$.  Then it is easy to check that $({\bf e}_i^T, {\bf e}_i^T), ({\bf e}_j^T, {\bf e}_j^T)$, and $({\bf e}_i^T+{\bf e}_j^T, {\bf e}_i^T+{\bf e}_j^T)$ are three different columns of $G$; therefore, the minimum distance of $(\mathcal{C}^{(2)}_{D^{*}})^{\bot}$ is 3.

By Proposition 5.3, $(\mathcal{C}^{(2)}_{D^{*}})^{\bot}$ has parameters $$[(2^{|A|}+2^{|B|}-2^{|A\cap B|})^{2}-1,(2^{|A|}+2^{|B|}-2^{|A\cap B|})^{2}-1-2|A\cup B|,3].$$ By Sphere Packing Bound, let $n=(2^{|A|}+2^{|B|}-2^{|A\cap B|})^{2}-1$, then it is easy to check that  $$\sum_{i=0}^{2}\dbinom{n}{i}=1+n+\frac{n(n-1)}{2}>2^{2|A\cup B|}.$$

This completes the proof.
\end{proof}

\begin{prop}{\rm

Let $\Delta $ be a simplicial  complex with two maximal elements $A, B\subseteq [m]$.  Let $D=\Delta+w\Delta \subset \Bbb F_4^m$.  Then the subfield code  $\mathcal{C}^{(2)}_{D^{c}}$  with respect to  $\mathcal{C}_{D^{c}}$ in Theorem 4.8 is 
a $$[4^{m}-(2^{|A|}+2^{|B|}-2^{|A\cap B|})^{2},2m,2^{2m-1}-(2^{|A|-1}+2^{|B|-1})( 2^{|A|}+2^{|B|}-2^{|A\cap B|})]$$ binary code and its weight distribution
 is presented in Table 12.  \begin{table}[h]
\caption{Weight distribution of the code in Proposition 5.6}
\begin{center}
\begin{tabu} to 1\textwidth{X[2,c]|X[1.8,c]}
\hline
Weight &Frequency\\
\hline
$0$&$1$\\
\hline
$2^{2m-1}-2^{|A|-1}( 2^{|A|}+2^{|B|}-2^{|A\cap B|})$&$2(2^{|A\backslash B|}-1)4^{m-|A\cup B|}$\\
\hline
$2^{2m-1}-2^{|B|-1}( 2^{|A|}+2^{|B|}-2^{|A\cap B|})$&$2(2^{|B\backslash A|}-1)4^{m-|A\cup B|}$\\
\hline
$2^{2m-1}-(2^{|A|-1}+2^{|B|-1})( 2^{|A|}+2^{|B|}-2^{|A\cap B|})$&$2(2^{|A\backslash B|}-1)(2^{|B\backslash A|}-1)4^{m-|A\cup B|}$\\
 \hline
$2^{2m-1}-2^{|B|-1}( 2^{|A|+1}+2^{|B|}-2^{|A\cap B|+1})$&$(2^{|B\backslash A|}-1)^{2}4^{m-|A\cup B|}$\\
\hline
$2^{2m-1}-2^{|A|-1}( 2^{|A|}+2^{|B|+1}-2^{|A\cap B|+1})$&$(2^{|A\backslash B|}-1)^{2}4^{m-|A\cup B|}$\\
\hline

$2^{2m-1}-(2^{|A|-1}+2^{|B|-1})( 2^{|A|}+2^{|B|}-2^{|A\cap B|+1})$&$(2^{|B\backslash A|}-1)^{2}(2^{|B\backslash A|}-1)^{2}4^{m-|A\cup B|}$\\
\hline

   $2^{2m-1}-\frac12(2^{|A|}+2^{|B|}-2^{|A\cap B|})^{2}+\frac12 (2^{|A|}-2^{|A\cap B|})(2^{|B|}-2^{|A\cap B|})$&$2(2^{|A\backslash B|}-1)(2^{|B\backslash A|}-1)4^{m-|A\cup B|}$\\
 \hline
 
 $2^{2m-1}-\frac12(2^{|A|}+2^{|B|}-2^{|A\cap B|})^{2}-\frac12 (2^{|A|}-2^{|A\cap B|})2^{|A\cap B|}$&$2(2^{|A\backslash B|}-1)(2^{|B\backslash A|}-1)^{2}4^{m-|A\cup B|}$\\
 \hline
 
  $2^{2m-1}-\frac12(2^{|A|}+2^{|B|}-2^{|A\cap B|})^{2}-\frac12 (2^{|B|}-2^{|A\cap B|})2^{|A\cap B|}$&$2(2^{|A\backslash B|}-1)^{2}(2^{|B\backslash A|}-1)4^{m-|A\cup B|}$\\
 \hline
 
  $2^{2m-1}-\frac12(2^{|A|}+2^{|B|}-2^{|A\cap B|})^{2}$&$(4^{|A\cup B|}-4^{|A\backslash B|+|B\backslash A|})4^{m-|A\cup B|}$\\
 \hline
  $2^{2m-1}$&$4^{m-|A\cup B|}-1$\\
 \hline
\end{tabu}
\end{center}
\end{table}

}
\end{prop}

\begin{proof} Note that  $D^{c}=(\Delta^{c}+w\Bbb F_{2}^{m})\cup (\Delta+w\Delta_{}^{c})$. By Theorem 3.2, $\mathcal{C}^{(2)}_{D^{c}}$ can be generated by 
$$\mathcal{C}^{(2)}_{D^{c}}=\{c_{D^{c}}(\boldsymbol{\alpha}, \boldsymbol{\beta}):\boldsymbol{\alpha}, \boldsymbol{\beta}\in \Bbb F_2^m\},$$ 
where $$c_{D^{c}}(\boldsymbol{\alpha}, \boldsymbol{\beta})=((\boldsymbol{\alpha}\cdot \mathbf d_{2}+\boldsymbol{\beta}\cdot (\mathbf d_{1}+\mathbf d_{2} )_{\mathbf d_{1}\in \Delta_{}^{c}, \mathbf d_{2}\in \Bbb F_{2}^{m}}| (\boldsymbol{\alpha}\cdot \mathbf f_{2}+\boldsymbol{\beta}\cdot (\mathbf f_{1}+\mathbf f_{2})_{\mathbf f_{1}\in \Delta_{}, \mathbf f_{2}\in \Delta_{}^{c}})).$$
Hence
\begin{eqnarray*} 
\mbox{wt}(c_{D^{c}}(\boldsymbol{\alpha}, \boldsymbol{\beta}))
&=&|D^{c}|-\sum_{\mathbf d_1\in  \Delta_{}^{c}}\sum_{\mathbf d_2\in \Bbb F_{2}^{m}}\frac12 \sum_{y\in\Bbb F_2} (-1)^{(\boldsymbol{\alpha}{ \mathbf d_2}+\boldsymbol{\beta}{( \mathbf d_{1}+ \mathbf d_2}))y}\\
&-&\sum_{\mathbf f_1\in  \Delta_{}}\sum_{\mathbf f_2\in \Delta_{}^{c}}\frac12 \sum_{z\in\Bbb F_2} (-1)^{(\boldsymbol{\alpha}{  \mathbf f_2}+\boldsymbol{\beta}{( \mathbf f_{1}+ \mathbf f_2}))z}\\
&=&\frac 1 2|D^{c}|-\frac 12(\sum_{\mathbf d_1\in \Delta_{A}^{c}} (-1)^{\boldsymbol{\beta} \mathbf d_1})(\sum_{\mathbf d_2\in \Bbb F_{2}^{m}}(-1)^{(\boldsymbol{\alpha}+\boldsymbol{\beta}){ \mathbf d_2}})\\
&-&\frac 12(\sum_{\mathbf f_1\in \Delta_{}} (-1)^{\boldsymbol{\beta} \mathbf f_1})(\sum_{\mathbf f_2\in  \Delta^{c}}(-1)^{(\boldsymbol{\alpha}+\boldsymbol{\beta}){  \mathbf f_2}})\\
&=&2^{2m-1}(1-\delta_{\bf 0,\boldsymbol{\beta}}\delta_{\bf 0,\boldsymbol{\alpha+\beta}})-\frac12 (2^{|A|}+2^{|B|}-2^{|A\cap B|})^{2}+\frac12 T_{\boldsymbol{\beta}}T_{\boldsymbol{\alpha}+\boldsymbol{\beta}},
\end{eqnarray*} 
where $T_{\mathbf {u}}=2^{|A|}\psi(\mathbf {u}|A)+2^{|B|}\psi(\mathbf {u}|B)-2^{|A\cap B|}\psi(\mathbf {u}|A\cap B)$. Then the result follows from Proposition 5.3. 
\end{proof}

\begin{rem}{\rm

By massive computation, weight distributions of the binary subfield codes of these quaternary codes  can be also determined in the case of $D=D_{1}+wD_{2}$, where $D_{1}$ is generated by two maximal elements $A, B\subseteq [m]$ and  $D_{2}$ is generated by two maximal elements $C, F\subseteq [m].$ 
}
\end{rem}

{\color{blue}
\begin{rem}{\rm It is noted that the subfield code $\mathcal{C}^{(2)}_{D^{c}}$ with respect to $\mathcal{C}_{D^{c}}$ in Proposition 5.6 has at  most  eight weights when the sets  $A,B$ have the same size.

%Note that when the sets $A,B$ have the same sizes, then the subfield code  $\mathcal{C}^{(2)}_{D^{c}}$ with respect to $\mathcal{C}_{D^{c}}$ in Proposition 5.6 has at  most  eight weights. 

 }
\end{rem}

}

The following is an example of Proposition 5.6. 

\begin{example} {\rm Let $m=4$.

(1) If $A=\{1,2\}$ and $B=\{2,3\}$, then the code $\mathcal{C}^{(2)}_{D^{c}}$ in Proposition 5.6   is a five-weight quaternary $[220, 8,104]$ linear code with weight enumerator $$1+8z^{104}+195z^{110}+20z^{112}+16z^{116}+16z^{118}.$$   In fact, the optimal binary linear code has parameter $[220, 8,109] $, according to \cite{G2}.

(2) If $A=\{1,2\}$ and $B=\{3,4\}$, then the code $\mathcal{C}^{(2)}_{D^{c}}$ in Proposition 5.6   is a five-weight quaternary $[207, 8,100]$ linear code with weight enumerator $$1+18z^{100}+108z^{102}+81z^{104}+36z^{108}+12z^{114}.$$   In fact, the optimal binary linear code has parameter $[207, 8,102] $, according to \cite{G2}.

}
\end{example}

%By the Griesmer bound, we have \begin{eqnarray*}&&\sum_{i=0}^{2m-1}\bigg\lceil {\frac{2^{2m-1}-2^{|A|+|B|-1}}{2^i}} \bigg\rceil=\sum_{i=0}^{2m-1}{\frac{2^{2m-1}}{2^i}} -\sum_{i=0}^{2m-1}\bigg\lfloor {\frac{2^{|A|+|B|-1}}{2^i}} \bigg\rfloor \\
%&&=(2^{2m}-1)-(2^{|A|+|B|}-1)=2^{2m}-2^{|A|+|B|}.
 %\end{eqnarray*}
% \section{Summary and }

{\color{blue}
\section{Code comparisons and concluding remarks }

To show significant advantages of our  codes,  in this section, we present two tables. 

In Table 13, we list recent works on linear codes over finite fields constructed from simplicial complexes for the convenience of the reader.  Compare with known results, some quaternary codes and their subfield codes obtained in this work have ﬂexible and new parameters. To the best of our knowledge, this is the first paper on the construction of linear codes over a non-prime field and  their subfield codes   by using simplicial complexes. 
%We also need to emphasize that our paper is first one to construct   by using simplicial complexes.

%Recently,   are used in the defining set construction of linear codes over finite fields for optimal or minimal codes.  We record the table here (Table 13) 

Table 14 presents optimal quaternary linear codes from two simplicial complexes $\Delta_{A}$ and $\Delta_{B}$ in Theorem 4.4. In Table 14, $*$ indicates that the corresponding  codes are optimal codes, and \lq\lq new\rq\rq~ are also indicated according to the current data base \cite{G}. In fact,  according to the current data base \cite{G} we  find as least 9 new optimal codes (lengths: 12, 48, 56, 60, 192, 224, 240, 248, 252);  even though their parameters are not new but they are inequivalent to currently best-known linear codes.  We confirmed those results by Magma. 
}

 \begin{table}
\caption{Linear codes constructed from simplicial complexes }
\begin{center}
\begin{tabu} to 1\textwidth{|X[0.9,c]|X[0.6,c]| X[1.4,c]|X[3,c]|X[0.8,c]|X[0.8,c]|X[0.9,c]|}
\hline
\tiny{Reference} &\tiny{$q$-Ary}&\tiny{Defining Set}&\tiny{$[n,k,d]$ Code}& \tiny{\#Weight}&\tiny{ Bound}& \tiny{Result}\\
\hline
%\tiny{[4]}&\tiny{$q$-ary}&\tiny{$[q^{2+1, 4, q^{2}-q}]$}&\tiny{6}&\tiny{Griesmer bound} &\tiny{ Thm. 1.1}\\
%\cline{2-6}
%&\tiny{$p$-ary}&\tiny{$[p^{2m}+1,3m+1, p^{2m-1}(p-1)-p^{m-1}]$}&\tiny{6}&\tiny{} &\tiny{ Thm. 4.6}\\
%\hline

\multirow{11}*{\tiny{[14]}}&\multirow{11}*{\tiny{binary}}&\tiny{ $\Delta^{*}$}&\tiny{$[2^{|A|}-1,|A|,2^{|A|-1}]$}&\tiny{1}&\tiny{Griesmer} &\multirow{2}*{\tiny{ Lem.7}}\\
\cline{3-6}
&&\multirow{3}*{\tiny{$\Bbb F_{2}^{m}\backslash \Delta$}}&\tiny{$[2^{m}-2^{|A|}, m, 2^{m-1}-2^{|A|-1}]$}&\tiny{2}&\tiny{Griesmer} &\\
\cline{4-7}
&&&\tiny{$[2^{m-1}, m,4]$}&\tiny{}&\tiny{Sphere Packing} &{\tiny{ Ex.10}}\\
\cline{4-7}
&&&\tiny{$[2^{m}-2\sum_{i=1}^{s}2^{|A_{i}-1|}+s-1,m, 2^{m-1}-\sum_{i=1}^{s}2^{|A_{i}|-1}]$}&\tiny{}&\tiny{Griesmer} &{\tiny{ Coro.21}}\\
\cline{3-7}

&&\multirow{2}*{\tiny{ $\Delta^{*}$} }&\tiny{$[2^{|A_{1}|+2^{|A_{2}|}}-2, |A_{1}\cup A_{2}|, 2^{|A_{1}|-1}]$}&\tiny{3}&\tiny{} &\multirow{2}*{\tiny{Lem.26}}\\
\cline{4-6}
&& &\tiny{$[2^{|A_{1}|+2^{|A_{2}|}}-2^{|A_{1}\cap A_{2}|}-1, |A_{1}\cup A_{2}|, 2^{|A_{1}|-1}]$}&\tiny{4}&\tiny{} &\\
\cline{3-7}

&&\multirow{2}*{\tiny{$\Bbb F_{2}^{m}\backslash \Delta$}}&\tiny{$[2^{m}-2^{|A_{1}|}-2^{|A_{2}|}+1, m, 2^{m-1}-2^{|A_{1}|-1}-2^{|A_{2}-1|}]$}&\tiny{3 or 4}&\tiny{Griesmer} &\multirow{2}*{\tiny{ Thm.27}}\\
\cline{4-6}
&&&\tiny{$[2^{m}-2^{|A_{1}|}-2^{|A_{2}|}+2^{|A_{1}\cap A_{2}|}-1,m,2^{m-1}-2^{|A_{1}|-1}-2^{|A_{2}|-1}]$}&\tiny{4 or  5}&\tiny{} &\\
\hline

\multirow{2}*{\tiny{[19]}}&\multirow{2}*{\tiny{binary}}&\multirow{3}*{\tiny{ $\Delta_{A}\backslash \Delta_{B}$}}&\tiny{$[2^{|A|}-2^{|B|},|A|,2^{|A|-1}-2^{|B|-1}]$}&\tiny{2}&\tiny{Griesmer} &{\tiny{ Thm.5}}\\
\cline{4-7}

&&&\tiny{$[2^{|A|}-2^{|B|},|A|,2^{|A|}-2^{|B|}-|A|, 3 \mbox{ or } 4]$}&\tiny{}&\tiny{} &{\tiny{Thm.6}}\\
\cline{3-7}

%&&\tiny{ $(\Delta_{A}\cup \Delta_{B})\backslash \{0\}$}&\tiny{$[2^{|A|}+2^{|B|}-2,2^{|A|}+2^{|B|}-2-|A|-|B|,3]$}&\tiny{}&\tiny{} &{\tiny{Thm.9}}\\
\hline

\multirow{3}*{\tiny{[23]}}&\multirow{3}*{\tiny{4-ary}}&{\tiny{ $\Bbb F_{4}^{m}\backslash \Delta_{A}+w \Delta_{B}$}}&\tiny{$[(2^{m}-2^{|A|})2^{|B|}, m,3(2^{m+|B|-2}-2^{|A|+|B|-2})]$}&\tiny{5}&\tiny{} &{\tiny{ Thm.3.1}}\\
\cline{3-7}

&&{\tiny{ $\Bbb F_{2}^{m}+w\Delta_{B}$}}&\tiny{$[2^{m+|B|},m,2^{m+|B|-1}]$}&\tiny{2}&\tiny{} &{\tiny{Coro.3.2}}\\
\cline{3-7}

&&{\tiny{ $\Bbb F_{2}^{m}\backslash \Delta_{A}+w\Bbb F_{2}^{m}$}}&\tiny{$[(2^{m}-2^{|A|})2^{m}, m,3(2^{2m-2}-2^{|A|+m-2})]$}&\tiny{2}&\tiny{Griesmer} &{\tiny{Coro.3.3}}\\

\hline

\multirow{5}*{\tiny{[12]}}&\multirow{5}*{\tiny{p-ary}}&\multirow{5}*{\tiny{ $\Bbb F_{p}^{m}\backslash \Delta$}}&\tiny{$[p^{m}-r-1,m, (p-1)p^{m-1}-r]$}&\tiny{2}&\tiny{Griesmer} &{\tiny{ Thm.4.1}}\\
\cline{4-7}

&&&\tiny{$[p^{m}-2r-2,m, (p-1)p^{m-1}-2r-1]$}&\tiny{4}&\tiny{Griesmer} &{\tiny{Thm.4.4}}\\
\cline{4-7}

&&&\tiny{$[p^{m}-3(r+1),m, (p-1)p^{m-1}-3r-2]$}&\tiny{5}&\tiny{Griesmer} &{\tiny{Thm.4.7}}\\
\cline{4-7}

&&&\tiny{$[p^{m}-(r+1)(p-1),m, (p-1)p^{m-1}-(r+1)p+2r+1]$}&\tiny{4}&\tiny{Griesmer} &{\tiny{Thm.4.11}}\\
\cline{4-7}

&&&\tiny{$[p^{m}-(r+1)(p-2),m, (p-1)p^{m-1}-(r+1)p+3r+1]$}&\tiny{5}&\tiny{Griesmer} &{\tiny{Thm.4.14}}\\

\hline

\multirow{10}*{\tiny{{This paper}}}&\multirow{4}*{\tiny{4-ary}}&\multirow{2}*{\tiny{ $ \Delta_{A}+w\Delta_{B}$}}&\tiny{$[2^{|A|+|B|}-1, |A\cup B|, 2^{|A|+|B|-1}]$}&\tiny{2}&\tiny{} &{\tiny{ Prop.4.2}}\\
\cline{4-7}

&&&\tiny{$[(2^{|A|}+2^{|B|}-2^{|A\cap B|})^{2}-1, |A\cup B|]$}&\tiny{$\le 10$}&\tiny{} &{\tiny{Prop.4.7}}\\
\cline{3-7}

&&\multirow{2}*{\tiny{ $\Bbb F_{4}^{m}\backslash \Delta_{A}+w\Delta_{B}$}}&\tiny{$[4^{m}-2^{|A|+|B|}, m, 3\times 2^{2m-2}-3\times 2^{|A|+|B|-2}]$}&\tiny{3}&\tiny{Griesmer} &{\tiny{Thm.4.4}}\\
\cline{4-7}

&&&\tiny{$[4^{m}-(2^{|A|}+2^{|B|}-2^{|A\cap B|})^{2}, m]$}&\tiny{$\le 11$}&\tiny{} &{\tiny{Thm.4.10}}\\
\cline{2-7}

&\multirow{5}*{\tiny{binary}}&\multirow{2}*{{\tiny{ $ \Delta_{A}+w\Delta_{B}$}}}&\tiny{$[2^{|A|+|B|}-1, |A|+|B|,2^{|A|+|B|-1}]$}&\tiny{1}&\tiny{} &{\tiny{Prop.5.1}}\\
\cline{4-7}

&&&\tiny{$[2^{2m}-2^{|A|+|B|}, 2m]$}&\tiny{2}&\tiny{Griesmer} &{\tiny{Thm.5.2}}\\

\cline{3-7}

&&\multirow{3}*{{\tiny{ $ \Delta+w\Delta$}}}&\tiny{$[(2^{|A|}+2^{|B|}-2^{|A\cap B|})^{2}-1,2|A\cup B|]$}&\tiny{$\le 10$}&\tiny{} &{\tiny{Prop.5.3}}\\
\cline{4-7}

&&&\tiny{$[(2^{|A|}+2^{|B|}-2^{|A\cap B|})^{2}-1,(2^{|A|}+2^{|B|}-2^{|A\cap B|})^{2}-1-2|A\cup B|,3]$}&\tiny{}&\tiny{Sphere Packing} &{\tiny{Thm.5.5}}\\
\cline{4-7}
&&&\tiny{$[4^{m}-(2^{|A|}+2^{|B|}-2^{|A\cap B|})^{2},2m]$}&\tiny{$\le 11$}&\tiny{} &{\tiny{Prop.5.6}}\\
\hline
\end{tabu}
\end{center}
\end{table}

 \begin{table}[h]
\caption{Optimal quaternary linear codes from Theorem 4.4}
\begin{center}
\begin{tabu} to 0.8\textwidth{|X[0.5,c]|X[0.8,c]| X[1,c]|X[1,c]|X[0.8,c]|}
\hline
\tiny{$m$} &\tiny{$A$}&\tiny{$B$}& \tiny{$[n,k,d]$ Code}&\tiny{ Remark}\\
\hline
\multirow{3}*{\tiny{2}}&\multirow{3}*{\tiny{$(1,0)$}}&\tiny{$(1,0)$}&\tiny{$[12,2,9]^{*}$}&\tiny{new} \\
\cline{3-5}
&&\tiny{$(0,1)$}&\tiny{$[12,2,9]^{*}$}&\tiny{} \\

\cline{3-5}

&&\tiny{$(1,1)$}&\tiny{$[8,2,6]^{*}$}&\tiny{} \\

\cline{1-5}

\multirow{8}*{\tiny{3}}&\multirow{3}*{\tiny{$(1,0,0)$}}&\tiny{$(1,0,0)$}&\tiny{$[60,3,45]^{*}$}&\tiny{new} \\
\cline{3-5}
&&\tiny{$(0,1,1)$}&\tiny{$[56,3,42]^{*}$}&\tiny{new} \\

\cline{3-5}

&&\tiny{$(1,1,1)$}&\tiny{$[48,3,36]^{*}$}&\tiny{} \\

\cline{2-5}

&\multirow{3}*{\tiny{$(0,1,1)$}}&\tiny{$(0,1,0)$}&\tiny{$[56, 3, 42]^{*}$}&\tiny{new} \\
\cline{3-5}

&&\tiny{$(0,1,1)$}&\tiny{$[48,3,36]^{*}$}&\tiny{new} \\
\cline{3-5}

&&\tiny{$(1,1,1)$}&\tiny{$[32,3,24]^{*}$}&\tiny{} \\
\cline{2-5}

&\multirow{2}*{\tiny{$(1,1,1)$}}&\tiny{$(0,1,0)$}&\tiny{$[48,3,36]^{*}$}&\tiny{} \\
\cline{3-5}

&&\tiny{$(0,1,1)$}&\tiny{$[32,3,24]^{*}$}&\tiny{} \\
\cline{3-5}

\hline

\multirow{11}*{\tiny{4}}&\multirow{3}*{\tiny{$(1,0,0,0)$}}&\tiny{$(1,0,0,0)$}&\tiny{$[252, 4, 189]^{*}$}&\tiny{new} \\
\cline{3-5}
&&\tiny{$(1,1,0,0)$}&\tiny{$[248, 4, 186]^{*}$}&\tiny{new} \\

\cline{3-5}

&&\tiny{$(1,1,1, 0)$}&\tiny{$[240, 4, 180]^{*}$}&\tiny{new} \\

\cline{2-5}

&\multirow{3}*{\tiny{$(0,0,1,1)$}}&\tiny{$(0,0,1,1)$}&\tiny{$[240, 4, 180]^{*}$}&\tiny{new} \\
\cline{3-5}

&&\tiny{$(0,1,1,1)$}&\tiny{$[224, 4, 168]^{*}$}&\tiny{new} \\
\cline{3-5}

&&\tiny{$(1,1,1,1)$}&\tiny{$[192, 4, 144]^{*}$}&\tiny{new} \\
\cline{2-5}

&\multirow{3}*{\tiny{$(1,1,1,0)$}}&\tiny{$(0,1,1,0)$}&\tiny{$[224, 4, 168]^{*}$}&\tiny{new} \\
\cline{3-5}

&&\tiny{$(0,1,0,0)$}&\tiny{$[240, 4, 180]^{*}$}&\tiny{new} \\
\cline{3-5}

&&\tiny{$(1,1,1,0)$}&\tiny{$[192, 4, 144]^{*}$}&\tiny{} \\
\cline{2-5}

&\multirow{2}*{\tiny{$(1,1,1,1)$}}&\tiny{$(0,1,0,0)$}&\tiny{$[224, 4, 168]^{*}$}&\tiny{new} \\
\cline{3-5}

&&\tiny{$(0,0,1,1)$}&\tiny{$[192, 4, 144]^{*}$}&\tiny{new} \\
\cline{3-5}

\hline
\end{tabu}
\end{center}
\end{table}

\

The main contributions of this paper are the following

\begin{itemize}
\item A general explicit relationship between quaternary linear codes and their binary subfield codes in terms of generator matrices  and defining sets (Theorem 3.2).

\item The determination of  weight distributions of four classes of quaternary codes when these simplicial complexes are all generated by a single maximal element  or two maximal elements (Propositions  4.2, 4.7, and Theorems  4.4, 4.10).

\item  The determination of  weight distributions of four classes of the subfield codes of  those quaternary codes (Propositions  5.1, 5.3, 5.6 and Theorem 5.2).

\item  Two infinite families of    optimal  linear codes  meeting the  Griesmer Bound (Theorems 4.4,  5.2) and a class of binary almost optimal linear codes with respect to Sphere Packing Bound (Theorem 5.5).

\item At least 9 new optimal quaternary linear codes (Table 14).
\end{itemize}

Very recently, Hyun {\em et al.} \cite{HKWY}  extended the construction of linear codes to posets.
It would be interesting to find more  optimal quaternary  codes by employing  posets.

On the other hand,    the quaternary  linear code in Example 3.4 is \begin{align*}\mathcal C=\bigg\{(0,0,0,0), (w,0,1+w,1), (1+w,0,1,w), (1,0,w,1+w),\\
(1+w,1+w,w,w), (1,1,1+w,1+w), (w,w,1,1),\\
(1,1+w,1,1+w), (w,1,w,1), (1+w,w,1+w,w),\\
(1+w,1,0,w), (1,w,0,1+w), (w,1+w,0,1),\\
(0,w,w,0), (0,1+w,1+w,0), (0,1,1,0)
\bigg\}.\end{align*}
It is easy to check that its binary  subfield subcode $$\mathcal C| \Bbb F_{2}=\mathcal C\cap \Bbb F_{2}^{4}=\{(0,0,0,0), (0,1,1,0)\}.$$  We just wonder that whether there is a direct way to compute the binary  subfield subcodes of these quaternary codes obtained in this paper.

%The following lemma developed by Aschikhmin and Barg \cite{AB}   is a useful criterion for a linear code to be minimal.

%\begin{lemma} {\rm
%A linear code $\mathcal{C}$ over $\Bbb F_q$ with minimum distance $w_{\min}$ is minimal provided that $w_{\min}/{w_{\max}}>{(q-1)}/q$, where $w_{\max}$ denotes the maximum nonzero Hamming weight in the code $\mathcal{C}$.}
% \end{lemma}

%\begin{cor} {\rm Let $A$ be a proper subset of $[m]$ and $D=\Delta_A^c+w\Bbb F_2^m\subset \Bbb F_4^m$ in Corollary 3.3. Then the code $\mathcal C_D$ is minimal.

%}
%\end{cor}

%\begin{proof} The result follows from Lemma 3.7 and  $$\frac{w_{\min}}{{w_{\max}}}=\frac{3\times2^{2m-2}-3\times2^{|A|+m-2}}{3\times2^{2m-2}-2^{|A|+m-1}}=1-\frac{1}{2(3\times 2^{m-|A|}-1)}>\frac{3}4.$$
%\end{proof}

\noindent{\bf{\large Acknowledgements.}} The authors would like to thank the reviewers and the Associate Editor for their valuable comments and suggestions to improve the quality of this article.
%%%%%%%%%%%%%%%%%%%%%%%%%%%%%%%%%%%%%%%


\begin{thebibliography}{99}

 % \bibitem{ADHK}  A. R. Anderson, C. Ding, T. Helleseth, T. Kl$\phi$ve, How to build robust shared control systems,   Des. Codes Cryptogr., 15:  111-124, 1998.
 
 
 
 
 
 














  \bibitem{AB}  A. Ashikhmin,  A. Barg, “Minimal vectors in linear codes,”  IEEE Trans. Inf. Theory, 44:  2010-2017, 1998.

 % \bibitem{CDY} C. Carlet, C. Ding, J. Yuan, Linear codes from highly nonlinear functions and their secret sharing schemes, IEEE Trans. Inf. Theory, 51: 2089-2102, 2005.

\bibitem{CH} S. Chang,  J. Y. Hyun,  “Linear codes from simplicial complexes,”  Des. Codes Cryptogr.,  86:  2167-2181, 2018.

\bibitem{D1} C. Ding, “Linear codes from some 2-designs,”  IEEE Trans. Inf. Theory, 61(6):  3265-3275, 2015.








  \bibitem{DH} C. Ding, Z. Heng,“The subﬁeld codes of ovoid codes,” IEEE Trans. Inf. Theory, 65(8): 4715-4729, 2019.

  %\bibitem{G} J. H. Griesmer, “A bound for error-correcting codes,” IBM J. Res.

%Develop., vol. 4, pp. 532–542, 1960.

 % \bibitem{DHKW} C. Ding, T. Helleseth, T. Kl$\phi$ve, X. Wang, A general construction of  Cartesian authentication codes, IEEE Trans. Inf. Theory, 53: 2229-2235, 2007.

% \bibitem{DN} C. Ding, H. Niederreiter, Cyclotomic linear codes of order 3, IEEE Trans. Inf. Theory, 53(6):  2274-2277, 2007.

 \bibitem{G2} M. Grassl,  Bounds on the minimum distance of linear codes.  http://www.codetables.de.

 \bibitem{G} J. H. Griesmer, “A bound for error correcting codes,”  IBM J. Res. Dev., 4: 532-542, 1960.





% \bibitem{HDZ}  Z. Heng,   C. Ding,  Z. Zhou, Minimal  linear codes over finite fields,  Finite Fields Appl.,  54: 176-196, 2018.



 
   \bibitem{HD1} Z. Heng, C. Ding, “The subﬁeld codes of hyperoval and conic codes,” Finite Fields Appl.,  56: 308–331, 2019.
  
 
  \bibitem{HD2}  Z. Heng, C. Ding, “The subfield codes of $[q+ 1, 2, q] $ MDS codes,”   arXiv:2008.00695, 2020.

  \bibitem{HDW} Z. Heng, C. Ding, W. Wang, “Optimal binary linear codes from maximal arcs,” IEEE Trans. Inf. Theory,  66(9): 5387-5394, 2020.
 
 
  \bibitem{HWD} Z. Heng, Q. Wang, C. Ding, “Two families of optimal linear codes and their subfield codes,” IEEE Trans. Inf. Theory, 66(11): 6872-6883, 2020.


% \bibitem{HY1}  Z. Heng, Q. Yue,  A class of binary linear codes with at most three weights,   IEEE Commun. Lett., 19: 1488-1491,  2015.

  \bibitem{HP}  W. C. Huffman, V. Pless, Fundamentals of Error-Correcting Codes, Cambridge University Press, Cambridge, 2003.
  
  \bibitem{HKN} J.Y. Hyun, H.K. Kim, M. Na, “Optimal non-projective linear codes constructed from down-sets,” Discrete Appl. Math., 254: 135-145, 2019.

  \bibitem{HKWY} J. Y. Hyun,  H. K. Kim, Y. Wu, Q. Yue, “Optimal  minimal linear codes from posets,”  Des. Codes Cryptogr., 88(12): 2475-2492, 2020.
  
  
  
  

   \bibitem{HLL} J. Y. Hyun,  J. Lee, Y. Lee, “Infinite families of optimal linear codes constructed from simplicial complexes,”   IEEE Trans. Inf. Theory, 66(11): 6762-6775,  2020.
   
   \bibitem{MS} F. J. MacWilliams and N. J. A. Sloane,
\textit{The Theory of Error-Correcting Codes,}  Amsterdam, The Netherlands: North Holland, 1977. 


 %  \bibitem{LYL}  C. Li, Q. Yue,   F. Li,  Weight distributions of cyclic codes with respect to pairwise coprime order elements, Finite Fields Appl.,  28: 94-114, 2014.

 % \bibitem{LM2}  H. Liu,  Y. Maouche, Several new classes of linear codes with few weights, Cryptogr. Commun., 11: 137-146, 2019.
 
  \bibitem{TWD}   C. Tang, Q. Wang, C. Ding,  “The subfield codes and subfield subcodes of a family of MDS codes,”   arXiv:2106.07840, 2021.
 
  \bibitem{WZ}   X. Wang, D. Zheng,  “The subfield codes of several classes of linear codes, ” Cryptogr. Commun., 12(6): 1111-1131, 2020.


 \bibitem{WZZ}   X. Wang, D. Zheng, Y.  Zhang, “A class of subfield codes of linear codes and their duals,”  Cryptogr. Commun., 13(1): 173-196,  2021.
 
\bibitem{WL} Y. Wu, Y. Lee, “Binary LCD codes and self-orthogonal codes via simplicial complexes,” IEEE Commun. Lett., 24(6): 1159-1162, 2020.

%  \bibitem{WH}  Y. Wu,  J. Y. Hyun, Few-weight codes over $\Bbb F_p + u\Bbb F_p$ associated with down sets and their distance optimal Gray image, Discrete Appl. Math., 283: 315-322, 2020.

 %  \bibitem{WZY}  Y. Wu, X. Zhu, Q. Yue, Optimal few-weight codes from simplicial complexes, IEEE Trans. Inf. Theory, 66(6): 3657-3663, 2020.

 \bibitem{X} C. Xiang, “It is indeed a fundamental construction of all linear codes,” arXiv:1610.06355, 2016.


 \bibitem{XY}   C. Xiang, W. Yin,  “Two families of subfield codes with a few weights,” Cryptogr. Commun.,  13(1): 117-127, 2021.


 % \bibitem{WYS}  Y. Wu, Q. Yue, X. Shi, At most three-weight binary linear codes from generalized Moisio’s exponential sums,
%Des. Codes Cryptogr., 87: 1927-1943, 2019.

 \bibitem{ZWL}   D. Zheng, X. Wang, Y. Li, M. Yuan, “Subfield codes of linear codes from perfect nonlinear functions and their duals,”   arXiv:2012.06105, 2020.


  %\bibitem{ZLTD}  Z. Zhou, X. Li, C. Tang, and C. Ding, Binary LCD codes and self-orthogonal codes from a generic construction, IEEE Trans. Inf. Theory,
 %65(1):  16-27,  2019.
 
 
  \bibitem{ZW} X. Zhu, Y. Wei,  “Few-weight quaternary codes via simplicial complexes,”   AIMS Math., 6(5): 5124-5132,  2021.


% \bibitem{M1} C. Ding, Z. Heng, and Z. Zhou, Minimal binary linear codes, IEEE Trans. Inf. Theory, vol.
%64, no. 10, pp. 6536-6545, 2018.

% \bibitem{M2} D. Bartoli andM. Bonini,Minimal linear codes in odd characteristic. IEEE Trans. Inf. Theory,
%65(7): 4152-4155, 2019.

% \bibitem{M3} M. Bonini, M. Borello, Minimal linear codes arising from blocking sets. arXiv preprint
%arXiv:1907.04626, 2019.

% \bibitem{M3}  G. Xu, L. Qu, Three classes of minimal linear codes over the finite fields of odd characteristic,
%IEEE Trans. Inf. Theory, vol. 65, no. 11, pp. 7067-7078, 2019.








% \bibitem{M4}  C. Tang, Y. Qiu, Q. Liao, Z. Zhou, Full characterization of minimal linear codes as cutting blocking sets, arxiv, 2019.









%\bibitem{xb} J. Xiong; J. Bao.
%Sharp regularity for elliptic systems associated with transmission problems.
%\textit{Potential Anal.} \textbf{39} (2013), no. 2, 169--194.

\end{thebibliography}
\end{document}